\documentclass[journal]{IEEEtran}
\usepackage{algorithm}
\usepackage{algorithmic}
\usepackage{amsmath}
\usepackage{amsfonts}
\usepackage{amsthm}
\usepackage{amssymb}
\usepackage{enumerate,enumitem}
\usepackage{bm}
\usepackage{gensymb}
\usepackage{graphicx}
\usepackage{float}
\usepackage{graphicx}
\usepackage[small,bf]{caption}
\usepackage{subcaption}
\usepackage{url}
\usepackage{color}
\usepackage{epstopdf}
\usepackage{cuted}
\usepackage{tikz}
\usetikzlibrary{shapes,arrows,calc,positioning}
\usepackage{pgfplots}
\pgfplotsset{compat=newest}
\usepackage{tabularx}
\usepackage{csquotes}
\usepackage{tabulary}
\usepackage{todonotes}
\usepackage{lipsum}
\usepackage{enumitem}
\usepackage{textcomp}
\usepackage{xcolor}
\usepackage{physics}
\usepackage{bbm}
\usepackage{booktabs}
\usepackage{multirow}
\usepackage{pgfplots}
\usepgfplotslibrary{groupplots}
\pgfplotsset{compat=1.15}
\definecolor{skyblue1}{rgb}{0.447,0.624,0.812}
\definecolor{scarletred1}{rgb}{0.937,0.161,0.161}
\usepackage{url}
\usepackage{empheq}
\usepackage[most]{tcolorbox}
\usepackage[export]{adjustbox}
\usepackage[hidelinks,hypertexnames=false]{hyperref}
\hypersetup{
	colorlinks,
	citecolor=blue,
	filecolor=black,
	linkcolor=blue,
	urlcolor=black}
\usepackage{physics}
\usepackage{cite}
\newcommand{\bhl}[1]{{\color{black} #1}}
\newtheorem{theorem}{Theorem}
\newtheorem{proposition}{Proposition}
\newtheorem{assumption}{Assumption}

\usepackage{xfrac}
\usepackage{nomencl}
\makenomenclature
\allowdisplaybreaks
\usetikzlibrary{external}
\usepgfplotslibrary{fillbetween}
\graphicspath{{figures/},{pdf/},{eps/},{png/}}

\definecolor{mnavy}{RGB}{12,44,86}
\definecolor{mblue}{RGB}{0,40,120}

\tcbset{colback=white, colframe=mblue, sharp corners, left=-9.25pt, right=-0.75pt, top=-5pt, bottom=0pt}

\begin{document}
\title{Adaptive Control of Distributed Energy Resources for Distribution Grid Voltage Stability}


\author{Daniel Arnold,~\IEEEmembership{Member,~IEEE,}
        Shammya Saha,~\IEEEmembership{Member,~IEEE,}
        Sy-Toan Ngo,~\IEEEmembership{Member,~IEEE,}
        Ciaran Roberts,~\IEEEmembership{Member,~IEEE,}
        Anna Scaglione,~\IEEEmembership{Fellow,~IEEE,}
        Nathan Johnson,~\IEEEmembership{Member,~IEEE,}
        Sean Peisert,~\IEEEmembership{Member,~IEEE,}
        David Pinney,~\IEEEmembership{Member,~IEEE}
\thanks{Daniel Arnold, Ciaran Roberts, Sy-Toan Ngo, and Sean Peisert are with the Lawrence Berkeley National Lab, Berkeley, USA, email: dbarnold@lbl.gov}
\thanks{Shammya Saha is with the Electric Power Research Institute, Knoxville, Tennessee, USA.}
\thanks{Anna Scaglione is with the Electrical and Computer Engineering Department, Cornell Tech, New York, New York, USA.}
\thanks{Nathan Johnson is with The Polytechnic School, Arizona State University, Mesa, Arizona.}
\thanks{David Pinney is with National Rural Electric Cooperative Association.}
\thanks{This research was supported in part by the  Cybersecurity, Energy Security, 
and Emergency Response (CESER), Cybersecurity for Energy Delivery Systems (CEDS) program of the U.S. Department of Energy via the Cybersecurity via Inverter-Grid Automatic Reconfiguration (CIGAR) project and the Supervisory Parameter Adjustment for Distribution Energy Storage (SPADES) project under contract DE-AC02-05CH11231. Any opinions, findings, conclusions, or recommendations expressed in this material are those of the authors and do not necessarily reflect those of the sponsors of this work.}
}


\maketitle



\begin{abstract}
    
Volt-VAR and Volt-Watt functionality in photovoltaic (PV) smart inverters provide mechanisms to ensure system voltage magnitudes and power factors remain within acceptable limits.  However, these control functions can become unstable, introducing oscillations in system voltages when not appropriately configured or maliciously altered during a cyberattack.  In the event that  Volt-VAR and Volt-Watt control functions in a portion of PV smart inverters in a distribution grid are unstable, the proposed adaptation scheme utilizes the remaining and stably-behaving PV smart inverters and other Distributed Energy Resources \bhl{to mitigate the effect of the instability}. The adaptation mechanism is entirely decentralized, model-free, communication-free, and requires virtually no external configuration. We provide a derivation of the adaptive control approach and validate the algorithm in experiments on the IEEE 37 \bhl{and 8500} node test feeders.  

\end{abstract}


\nomenclature[01]{$v_{i}$}{Voltage magnitude at node $i$}
\nomenclature[02]{$\mathbf{Z}$}{Matrix collection of network resistance and reactance matrices, $\in \mathcal{R}^{n \times 2n}$}
\nomenclature[03]{$f_{p,i}$}{Volt-Watt curve for node $i$ smart inverter}
\nomenclature[04]{$f_{q,i}$}{Volt-VAR curve for node $i$ smart inverter}
\nomenclature[05]{$C_{p,i}$}{Lipschitz constant for node $i$ Volt-Watt curve}
\nomenclature[06]{$C_{q,i}$}{Lipschitz constant for node $i$ Volt-VAR curve}
\nomenclature[07]{$\mathbf{C}_{s}$}{Matrix collection of Lipschitz constants for Volt-Watt and Volt-VAR functions, $\in \mathcal{R}^{2n \times n}$}
\nomenclature[08]{$\mathbf{s}$}{Vector collection of nodal active and reactive powers, respectively, generated by DER, $\in \mathcal{R}^{2n \times 1}$}
\nomenclature[09]{$\xi_{i}$}{Low pass filtered node $i$ voltage magnitude}
\nomenclature[10]{$u_{p,i}$}{Active power injected by node $i$ DER capable of adaptive power injection}
\nomenclature[11]{$u_{q,i}$}{Reactive power injected by node $i$ DER capable of adaptive power injection}
\nomenclature[12]{$w_{i}$}{Adaptive voltage bias applied to node $i$ smart inverter Volt-VAR and Volt-Watt functions}

\printnomenclature


\section{Introduction}
\label{sec:introduction}
Increasing adoption of Distributed Energy Resources (DER), specifically rooftop photovoltaic (PV) generation systems, is challenging many conventionally-held models and practices regarding the operation of the electric power system.  While the presence of DER gives individuals and communities the ability to self-generate at least a portion of their load, they also make proper management of the power system more difficult as many DER are not utility-owned/operated.

Emerging standards \cite{IEEE_1547, rule21, inverter2016} are encouraging the use of device-level modulation of active and reactive power injection in response to local grid conditions.  These autonomous control functions allow DER to quickly correct undesirable voltages and power factors at the point of injection and (in theory) alleviate the need for a response from the grid managing entity.

Although these autonomous control functions, specifically smart inverter Volt-VAR (VV) and Volt-Watt (VW) controllers, are well-intentioned, numerous works have emerged showing that proper configuration of individual devices is crucial for the stable operation of the DER population.  Jahangiri et al. \cite{jahangiri2013distributed} discussed the phenomenon of \textquote{hunting} in voltages in systems with VV control. Farivar et al. \cite{farivar2013equilibrium} modeled the interaction between system voltage magnitudes and PV inverter VV functions as a feedback control loop which explicitly tied the slopes of VV controllers of inverters to unstable (highly oscillatory) reactive power injections.  Although the instability threshold depends on the network's specifics, instability is reached when the slopes of the VV control curves become too steep. Numerous other works have also modeled the inverter/grid interaction as a first-order feedback controller and arrived at similar stability conditions \cite{zhou2016local, braslavsky2017voltage, bakerNetwork2017, eggli2020stability, saha2020lyapunov}. Moreover, adaptive control approaches have been previously applied to improve the interaction of PV systems and the electric grid. For instance, Ghasemi et al. \cite{GHASEMI2016prevention} considered a control law to adapt PV reactive power injection for loss minimization and over-voltage prevention.  Furthermore, the standards for smart inverter functionality have changed and expanded rapidly in the past few years (e.g., IEEE 1547-2018 \cite{IEEE_1547}), which has led to a variety of different control interfaces and debate in the industry over best practices for managing smart inverter functionality.

While instabilities may arise naturally in a system (due to system reconfiguration, poor parameterizations of VV/VW functions, or intermittency in PV penetration), the remote update capability of many smart inverter devices presents a vulnerability that a malicious entity could purposefully exploit to destabilize the smart inverter/grid interaction \cite{sahoo2019cyber}. Security researchers have identified exploitable vulnerabilities in deployed inverter firmware \cite{practical2017}, foreign nations are actively targeting the US bulk power system \cite{russian2018}, and in at least one instance, a US inverter control system has been successfully attacked \cite{risks2019}.  An excellent example of the extent to which aggregations of smart inverters can be remotely updated was illustrated in Hawaii, where local utilities worked with a smart inverter vendor to remotely update the autonomous control functions of 800,000 inverters in a single day \cite{spectrum2015}.  

\subsection{Contributions}
\begin{figure}[]
    \centering
    \includegraphics[scale=0.99]{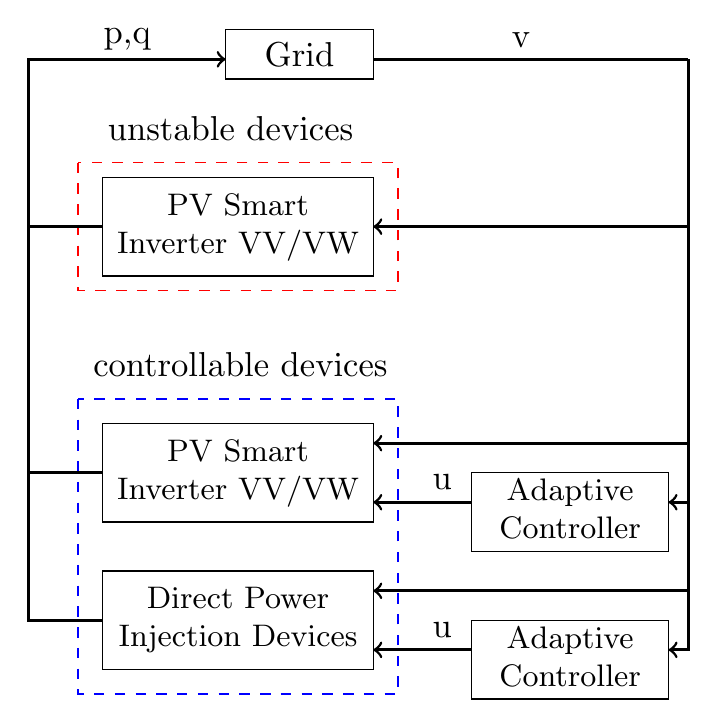}
    \caption{Block diagram of controllable devices.}
    \label{fig:devices_block_diagram}
\end{figure}
In this work, we propose an adaptive control approach to adjust the VV/VW control functions of stable PV systems (which we will henceforth refer to as \emph{non-compromised}) and active/reactive power injections of other DER (e.g., battery storage systems) to mitigate voltage instabilities in the system. We utilize a Model Reference Adaptive Controller (MRAC) approach \cite{astrom2008adaptive} to derive a stabilizing control law.  As MRACs utilize a stable \emph{reference model} to compare to the unstable plant, our work employs the low pass filtered AC grid voltage magnitude measured by the PV system/DER as a proxy for the stable reference model.  This particular choice of reference model makes the algorithm decentralized, (system) model-free, and communications-free. The proposed adaptation mechanism is designed via Lyapunov analysis to drive the error between the plant and the reference asymptotically to zero.   

The proposed adaptation law complies with existing smart inverter standards \cite{IEEE_1547,rule21, inverter2016}, allowing for a decentralized implementation directly within the smart inverter or DER with minimum additional investment.  The control law generalizes to both controllable PV smart inverter systems and devices capable of direct power injection (the adaptation law is identical for both device classes). Fig. \ref{fig:devices_block_diagram} shows a high-level diagram of the role of the adaptive control approach, illustrating three categories of devices in feedback with the electric grid. All devices can inject active ($p$) and reactive ($q$) power into the grid and can measure the voltage ($v$) at their point of interconnection.  While a portion of smart inverters in the system with VV/VW capabilities are unstable, another portion of smart inverters and other DER capable of directly injecting power into the grid are adjusted by the proposed adaptive control signal $u$. For the purposes of this work, direct power injection devices are systems that do not determine their power injection setpoints indirectly through another control mechanism such as VV/VW control.  An example of a direct power injection device is a battery storage system capable of supporting direct charge/discharge requests \cite{inverter2016}.  For devices such as these, $u$ is the additional amount of active/reactive power injected (or consumed, if $u$ changes sign) into the grid. However, for PV smart inverter VV/VW controllers (or both), $u$ is a voltage bias added to the measured grid voltage magnitude.

\bhl{In a related work by Singhal et al. \cite{singhal2019real}, the authors adopted a similar strategy of translating the smart inverter VV curves for steady-state error reduction and flattened VV curves to mitigate instabilities.  While stretching/flattening VV curves proved effective in reducing oscillations in situations considered in their analysis, the strategy has two drawbacks: 1) it does not account for the scenario where a subset of the PV devices in a given system are solely responsible for creating instabilities and are not controllable, and 2) the algorithm requires the use of \textquote{correction factors} which are estimated offline based on sensitivity analysis and engineering judgment.  These correction factors would need to be estimated by an entity with complete grid knowledge and relayed to individual smart inverters, requiring substantial communications overhead that creates a cyber-security vulnerability.  Additionally, correction factor computation would need to be carried out periodically as new PV systems are installed, or the grid is reconfigured.}  Note that we explicitly assume that unstable PV smart inverter systems are not controllable as they may have been, in the worst case, compromised via a cyber-attack and should be considered unreliable for control.  

\bhl{
To summarize, the contribution of this paper is the development of an adaptive control scheme to mitigate inverter-driven oscillations caused by a portion of DER smart inverters with unstable VV/VW settings.  The adaptive control scheme has the following properties:
\begin{enumerate}
    \item The approach is model-free and requires no knowledge of the topology of the system.
    \item The approach utilizes \emph{non-compromised} DER to mitigate oscillations introduced by other smart inverters.
    \item The adaptation law complies with existing smart inverter standards (IEEE 1547 \cite{IEEE_1547}).
    \item The scheme is capable of mitigating unstable inverter-driven oscillations in the seconds after the oscillations first manifest in the network.
\end{enumerate}
}
The paper is organized as follows. In Section \ref{sec:preliminaries} we briefly derive a linear balanced power flow model for subsequent control design and discuss smart inverter dynamic modeling assumptions. Derivation of the adaptive control scheme is presented in Section \ref{sec:control_design}.  Simulation results showing the controller's effectiveness through time series simulations of three-phase \emph{unbalanced} distribution systems are discussed in Section \ref{sec:results}.  Finally, concluding remarks are presented in Section \ref{sec:conclusions}.

\section{Preliminaries}
\label{sec:preliminaries}

\bhl{Unless otherwise stated, upper case (lower case) boldface letters will be used to denote matrices (column vectors).  $\vert \mathbf{x} \vert$ denotes the element-wise absolute value of the vector $\mathbf{x}$. The operation $\text{diag}(\mathbf{x})$, where $\mathbf{x} \in \mathcal{R}^{n\times 1}$, returns an $n \times n$ matrix with the elements of $\mathbf{x}$ along the main diagonal.}

\subsection{Linear Power Flow Model Derivation}
Let the graph $\mathcal{G} = (\mathcal{N} \cup \{0\},\mathcal{L})$ represent a balanced radial distribution feeder, where $\mathcal{N}$ is the set of nodes (excluding the substation) and $\mathcal{L}$ is the set of line segments, where $|\mathcal{N}| = |\mathcal{L}| = n$.  For a given bus $i \in \mathcal{N}$, let $\mathcal{L}_{i}$ (where $\mathcal{L}_{i} \subseteq \mathcal{L}$) denote the collection of line segments from node 0 (e.g. the substation) to node $i$. The \emph{DistFlow} equations \cite{baran1989optimal} capture the relationship between power flowing in line segment $(i,j) \in \mathcal{L}$ and the voltage magnitude drop between nodes $i$ and $j$:

\begin{subequations}
\begin{gather}
    P_{ij} = p_{j}^{c} - p_{j}^{g} + r_{ij}c_{ij} + \sum_{k:(j,k) \in \mathcal{L}} P_{jk} \label{eq:Pij} \\ 
    Q_{ij} = q_{j}^{c} - q_{j}^{g} + x_{ij}c_{ij} + \sum_{k:(j,k) \in \mathcal{L}} Q_{jk} \label{eq:Qij}\\
    v_{j}^{2} - v_{i}^{2} = -2 \big( r_{ij}P_{ij} + x_{ij}Q_{ij}\big) + \big( r_{ij}^{2} + x_{ij}^{2}\big)c_{ij}^{2} \label{eq:v_drop},
\end{gather}
\end{subequations}

\noindent where $v_{i}^{2}$ is node $i$ squared voltage magnitude, $P_{ij}$ and $Q_{ij}$ denote the active/reactive power flowing in line segment $(i,j)$, $r_{ij}$ and $x_{ij}$ are line segment $(i,j)$ resistance and reactance, and $c_{ij}$ are losses.  For node $i$, active (reactive) power consumption is denoted by $p^{c}_{i}$ ($q_{i}^{c}$) and active (reactive) power generation, due to DER, is denoted by $p^{g}_{i}$ ($q_{i}^{g}$).  

Consistent with \cite{baran1989network,farivar2013equilibrium}, we neglect losses in \eqref{eq:Pij} - \eqref{eq:v_drop} which is achieved via setting $c_{ij} = 0$ for all  $(i,j) \in \mathcal{L}$.  Furthermore, as $v_{i} \approx 1$ we approximate $v_{j}^{2} - v_{i}^{2} \approx 2(v_{j} - v_{i})$.  Let $\beta(j)$ denote the set of all nodes descended from $j$ (including $j$ itself).  With these changes, the \emph{DistFlow} model becomes:
\begin{subequations}
\begin{gather}
    P_{ij} = \sum_{k \in \beta(j)}\big( p_{k}^{c} - p_{k}^{g} \big) \label{eq:P_lin} \\
    Q_{ij} = \sum_{k \in \beta(j)}\big( q_{k}^{c} - q_{k}^{g} \big) \label{eq:Q_lin}\\
    v_{i} - v_{j} = r_{ij}P_{ij} + x_{ij}Q_{ij}. \label{eq:v_lin}
\end{gather}
\end{subequations}

The now linearized system of \eqref{eq:P_lin} - \eqref{eq:v_lin} can be more compactly represented via substituting \eqref{eq:P_lin} and \eqref{eq:Q_lin} into \eqref{eq:v_lin} and making successive substitutions of voltages from upstream nodes yielding node $i$ voltage as a function of feeder head voltage $v_{0}$.  If one defines the following vectors:
\begin{subequations}
\begin{gather}
    \mathbf{v} = [v_{1}, \ldots, v_{n}]^\top, \quad \mathbf{v_{0}} = v_{0}\mathbf{1} \\
    \mathbf{p^{c}} = [p_{1}^{c}, \ldots, p_{n}^{c}]^\top, \quad \mathbf{p^{g}} = [p_{1}^{g}, \ldots, p_{n}^{g}]^\top \\
    \mathbf{q^{c}} = [q_{1}^{c}, \ldots, q_{n}^{c}]^\top, \quad \mathbf{q^{g}} = [q_{1}^{g}, \ldots, q_{n}^{g}]^\top,
\end{gather}
\end{subequations}

\noindent then the system of \eqref{eq:P_lin} - \eqref{eq:v_lin} can be recast in vector form:
\begin{gather}
    \mathbf{v} = \mathbf{v_{0}} + \mathbf{R}\big(\mathbf{p^{g}} - \mathbf{p^{c}}\big) + \mathbf{X}\big(\mathbf{q^{g}} - \mathbf{q^{c}}\big) \label{eq:lindistflow},
\end{gather}

\noindent where $\mathbf{R}$ and $\mathbf{X}$ are completely positive matrices \cite{farivar2013equilibrium} and 
\begin{subequations}
\begin{gather}
    R_{ij} = \sum_{(h,k) \in \mathcal{L}_{i} \bigcap \mathcal{L}_{j}} r_{hk} \label{eq:Rij} \\ X_{ij} = \sum_{(h,k) \in \mathcal{L}_{i} \bigcap \mathcal{L}_{j}} x_{hk}\label{eq:Xij}.
\end{gather}
\end{subequations}

Defining $\mathbf{Z} = [\mathbf{R}, \mathbf{X}]$, $\mathbf{s}^{c} = [\mathbf{p}^{c},\mathbf{q}^{c}]^\top$, and $\mathbf{s}^{g} = [\mathbf{p}^{g},\mathbf{q}^{g}]^\top$, \eqref{eq:lindistflow} can expressed compactly as:
\begin{gather}
    \mathbf{v} = \mathbf{v_{0}} + \mathbf{Z}\big(\mathbf{s}^{g} - \mathbf{s^{c}}\big) \label{eq:lindistflow2},
\end{gather}

\subsection{Smart Inverter Models}

Smart inverter VV and VW functions compute reactive and active power set-points, respectively, as functions of deviations of locally sensed voltages from a nominal value (typically $1$ p.u.).  Let $f_{p,i}(v_{i})$ and $f_{q,i}(v_{i})$ denote the VV and VW control functions for a smart inverter at node $i$.  \bhl{These control laws are depicted in Figs. \ref{fig:vvc} - \ref{fig:vwc}, respectively, and consist of continuous piece-wise linear functions of the voltage deviation $v_{i} - v_{\text{nom}}$.  The functions are parameterized by the vector $\bm{\eta} = [\eta_{1}, \dots, \eta_{5}]$, which can be used to alter the location and slopes of non-zero segments of the controllers.}  We make the following assumptions regarding these functions \cite{farivar2013equilibrium, farivar2015local,saha2020lyapunov}:

\begin{assumption}
\label{assumption:f_1}
The functions $f_{p,i}(v_{i})$ and $f_{q,i}(v_{i})$ are monotonically decreasing and continuously piece-wise differentiable.
\end{assumption}

\begin{assumption}
\label{assumption:f_2}
Both $f_{p,i}(v_{i})$ and $f_{q,i}(v_{i})$ have bounded derivatives, i.e. there exists $C_{p,i} < +\infty$ and $C_{q,i} < +\infty$ such that $\lvert f^{\prime}_{p,i}(v_{i}) \rvert \le C_{p,i}$ and $\lvert f^{\prime}_{q,i}(v_{i}) \rvert \le C_{q,i}$ for all $v_{i}$.
\end{assumption}

Let $\bar{s}_{i}$ denote the rated apparent power (i.e. the inverter capacity) of the smart inverter at node $i$.  Similarly, let $\bar{p}_{i}$ denote the maximum available active power capable of being sourced at the present irradiance level.  Following the analysis of \cite{braslavsky2017voltage}, $\bar{p}_{i}$ can be expressed as a fraction of the capacity of the $i^{th}$ inverter:
\begin{equation}
    \bar{p}_{i} = \lambda \bar{s}_{i}, \quad 0 < \lambda \le 1, \label{eq:pbar} 
\end{equation}

\noindent where $\lambda = 1$ corresponds to the inverter generating the maximum amount of active power.  In situations where the amount of active power generated is less than $\bar{s}_{i}$, some inverter devices support the use of the excess system capacity for reactive power generation.  The maximum amount of reactive power available for injection/consumption, denoted by $\bar{q}_{i}(v_{i})$, is a function of hardware limitations ($q_{i}^{\text{lim}}$) and the available reactive power \cite{saha2020lyapunov}:
\begin{equation}
    \bar{q}_{i}(v_{i}) = \min\Bigg(q_{i}^{\text{lim}}, \sqrt{\bar{s}_{i}^{2} - f_{p,i}^{2}(v_{i})}\Bigg). \label{eq:qbar}
\end{equation}

A derivation of the Lipschitz constants for both the VV and VW control functions depicted in  Figs. \ref{fig:vvc} - \ref{fig:vwc} can be found in \cite{saha2020lyapunov}.  

\begin{figure}[h!]
\centering
\includegraphics[width=1\columnwidth]{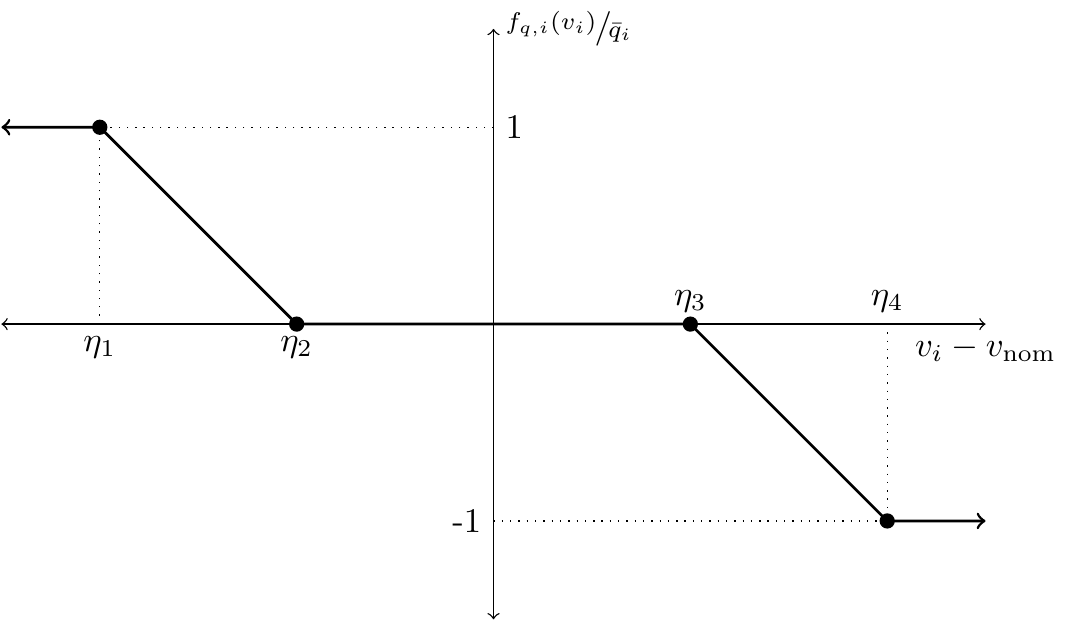}
\caption{Inverter Volt-VAR curve.  Positive values denote VAR injection.  $v_{\text{nom}}$ is the nominal voltage value.}
\label{fig:vvc}
\end{figure}

\begin{figure}[h!]
\centering

\includegraphics[width=1\columnwidth]{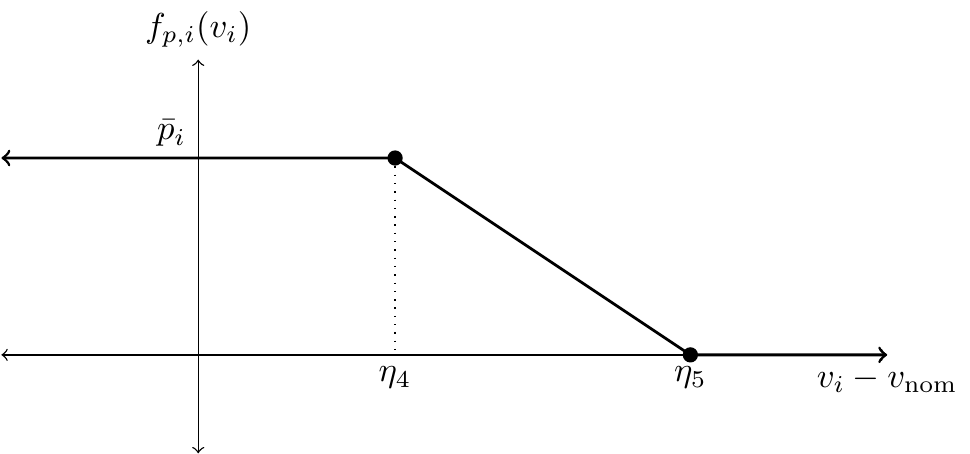}
\caption{Inverter Volt-Watt curve.  Positive values denote watt injection.  $v_{\text{nom}}$ is the nominal voltage value.}
\label{fig:vwc}
\end{figure}

A block diagram of the smart inverter model considered in this work is shown in Fig. \ref{fig:inverter_block_diagram} with VV and VW control logic.  As is shown in the figure, the grid voltage $v$ is the input to the VV and VW controllers.  The maximum available active power from the solar array, $\bar{p}$, is also input into the VW controller, which along with $v$, determines the maximum amount of reactive power available for injection/consumption $\bar{q}$ that is then input to the VV controller.  The active and reactive power setpoints produced by the VW and VW controllers are then low pass filtered by $H_{O}(s)$ to produce the active and reactive power injections that are injected into the grid. These filters serve to limit the rate at which the active and reactive powers injected by PV systems can change and do not represent physical constraints of the smart inverter devices themselves \cite{inverter2016}.


\begin{figure}[h]
\includegraphics[width=0.48\textwidth]{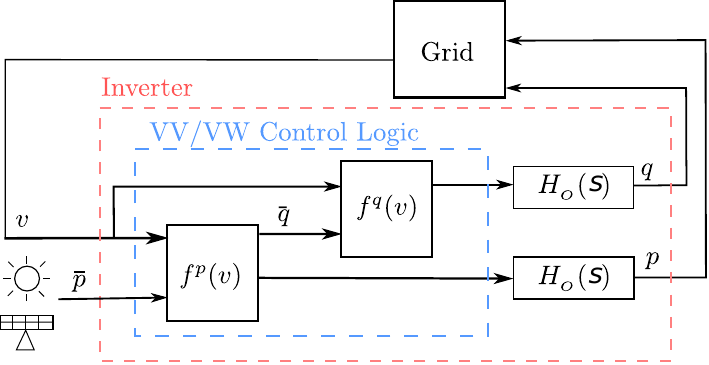}
\caption{Block diagram of VV and VW control logic of an inverter.}
\label{fig:inverter_block_diagram}
\end{figure}

\bhl{
\subsection{Smart Inverter Dynamics}
\label{subsec:stability}

In this section we develop a dynamic model of smart inverters, modeled by Fig. \ref{fig:inverter_block_diagram}, connected to the distribution grid.  Without loss of generality, we assume the presence of a VV and VW capable smart inverter at each node in the system.  To begin, let $\mathbf{f}(\mathbf{v}) = [\mathbf{f}_{p}(\mathbf{v}), \mathbf{f}_{q}(\mathbf{v})]^\top$ denote the collection of inverter VV and VW functions at each node in $\mathcal{G}$, where:
\begin{subequations}
\begin{align}
    \mathbf{f}_{p}( \mathbf{v}) &= [f_{p,1}(v_{1}), \ldots, f_{p,n}(v_{n})]^\top \label{eq:fpi}\\
    \mathbf{f}_{q}( \mathbf{v}) &= [f_{q,1}(v_{1}), \ldots, f_{q,n}(v_{n})]^\top \label{eq:fqi},
\end{align}
\end{subequations}

\noindent where, according to Assumptions \ref{assumption:f_1}-\ref{assumption:f_2}, both $f_{p,i}(v_{i})$ and $f_{q,i}(v_{i})$ are locally Lipschitz with constants $C_{p,i}$ and $C_{q,i}$, respectively.  Define the matrices
\begin{subequations}
    \begin{align}
        \mathbf{C}_{p} &= \text{diag}([C_{p,1}, \ldots, C_{p,n}])\\
        \mathbf{C}_{q} &= \text{diag}([C_{q,1}, \ldots, C_{q,n}])\\
        \mathbf{C}_{s} &= \mqty[ \mathbf{C}_{p} & \mathbf{C}_{q}]^\top.
    \end{align}
\end{subequations}

Under the additional assumption that active and reactive power consumption due to system loads change slowly with respect to inverter control actions, \eqref{eq:lindistflow2} can be recast in the following form:
\begin{align}
    \mathbf{v} = \mathbf{Z}\mathbf{s} + \underbrace{\mathbf{v_{0}} - \mathbf{Z}\mathbf{s^{c}}}_{\mathbf{\bar{v}}} \label{eq:lindistflow3},
\end{align}
\noindent where the superscript has been dropped from $\mathbf{s}$ for convenience and $\mathbf{\bar{v}}$ is treated as constant. Consistent with Fig. \ref{fig:inverter_block_diagram}, the dynamics of the inverter consist of nonlinear Volt-VAR \& Volt-Watt controllers in series with first order low pass filters can be expressed as \cite{saha2020lyapunov}:
\begin{subequations}
\begin{align}
    \mathbf{T}\mathbf{\dot{s}} &= \mathbf{f}(\mathbf{v}) -\mathbf{s}, \label{eq:sdot}\\
    \mathbf{v} &= \mathbf{Z}\mathbf{s} + \mathbf{\bar{v}} \label{eq:v}
\end{align}
\end{subequations}
\noindent where $\mathbf{T} \in \mathbb{R}^{2n \times 2n}$ is a diagonal and positive definite matrix that collects low pass filter time constants.  Substituting \eqref{eq:v} into \eqref{eq:sdot} yields the desired dynamics in terms of the state variable $\mathbf{s}$:
\begin{equation}
    \mathbf{T}\mathbf{\dot{s}} = \mathbf{f}(\mathbf{Z}\mathbf{s} + \mathbf{\bar{v}}) -\mathbf{s}. \label{eq:sdot2}
\end{equation}

A derivation of the stability criterion of \eqref{eq:sdot2} in terms of the system impedances and smart inverter Lipschitz constants can be found in Appendix \ref{subsec:appendix_stability}.
}

\section{Adaptive Control Design}
\label{sec:control_design}
The proposed adaptive control approach uses non-compromised (i.e. stably-behaving) devices to drive system voltages to regions where the compromised smart inverter VV/VW controllers produce constant power with respect to changing voltages (i.e., the flat regions of Figs. \ref{fig:vvc} - \ref{fig:vwc}).  Two types of devices are considered: 1) devices capable of direct power injection/consumption, 2) and non-compromised PV systems (see Fig. \ref{fig:devices_block_diagram}).
\subsection{Adaptive (Direct) Power Injection/Consumption}
\label{subsec:direct_power_injection}
In this section, we consider the effect of adaptive power injection/consumption (both active and reactive power) as a means to stabilize \eqref{eq:sdot} - \eqref{eq:v}, in the event that the criteria of Proposition \ref{prop:direct} is violated.  Devices capable of providing this capability include many types of DER (e.g., battery storage, PV systems, and electric vehicles), and conventional loads (e.g., demand response).  The goal in this section is to derive a decentralized control law to determine how much active/reactive power should be injected (or consumed) by individual devices in order to achieve stability.  

To begin, consider $l$ devices capable of directly injecting \emph{active} power which are deployed at a subset of nodes in $\mathcal{G}$ (e.g., $l \le n$).  Let $\mathbf{u}_{p} \in \mathcal{R}^{l \times 1}$ represent the collection of these controllers.  The controllers collected in $\mathbf{u}_{p}$ are mapped to nodal locations in $\mathcal{G}$ via the sparse matrix $\mathbf{B}_{p} \in \mathcal{R}^{n \times l}$ where $\mathbf{B}_{p}(i,j) = 1$ indicates the controller at entry $j$ in $\mathbf{u}_{p}$ is located at node $i$ in $\mathcal{G}$.  As such, each row and column of $\mathbf{B}_{p}$ consists of all $0$s with at most a single entry equal to 1.

Additionally, consider $m$ devices capable of directly injecting \emph{reactive} power which are individually deployed at a subset of nodes in $\mathcal{G}$ (e.g., $m \le n$). Let $\mathbf{u}_{q} \in \mathcal{R}^{m\times 1}$ represent the collection of these controllers.  Similarly to the active power case, controllers in $\mathbf{u}_{q}$ are mapped to nodal locations in $\mathcal{G}$ via the sparse matrix $\mathbf{B}_{q} \in \mathcal{R}^{n \times m}$ where $\mathbf{B}_{q}(i,j) = 1$ indicates the controller at entry $j$ in $\mathbf{u}_{q}$ is located at node $i$ in $\mathcal{G}$.  As such, each row and column of $\mathbf{B}_{q}$ consists of all $0$s with at most a single entry equal to 1.  Now, define:
\begin{equation}
    \mathbf{u} = [\mathbf{u}_{p},\mathbf{u}_{q}]^\top, \quad \mathbf{B} = \mqty[ \mathbf{B}_{p} & \mathbf{0} \\ \mathbf{0} & \mathbf{B}_{q} ], \label{eq:u_and_B}
\end{equation}

\noindent where $\mathbf{u} \in \mathcal{R}^{(l+m)\times 1}$ and $\mathbf{B} \in \mathcal{R}^{2n \times (l+m)}$.  With these definitions, the effect of direct power injection control on the dynamics of \eqref{eq:sdot} - \eqref{eq:v} can be expressed as:
\begin{gather}
    \mathbf{T}\mathbf{\dot{s}} = \mathbf{f}(\mathbf{Z}\mathbf{s} + \mathbf{\bar{v}} + c\mathbf{B}\mathbf{u}) -\mathbf{s}, \label{eq:sdot_dpi}
\end{gather}

\noindent which has equilibrium:
\begin{equation}
        \mathbf{0} = \mathbf{f}(\mathbf{Z}\mathbf{s}^{*} + \mathbf{\bar{v}} + c\mathbf{B}\mathbf{u}^{*}) -\mathbf{s}^{*} \label{eq:eq_dpi}.
\end{equation}

Here the parameter $c \in \{-1,1\}$ indicates whether power is being consumed ($c=-1$) or injected ($c=1$).  Note that this formulation implicitly assumes that all devices participating in this control activity either inject or consume power (i.e., the sign of the elements of $\mathbf{u}$ must be the same).  The following assumptions are made regarding the forced and unforced equilibria of \eqref{eq:sdot_dpi}:

\begin{assumption}
\label{assumption:direct1}
The equilibrium of the unforced system of \eqref{eq:sdot_dpi} (Eq.\eqref{eq:eq_dpi} where $\mathbf{u}^{*} = \mathbf{0}$) is unstable.
\end{assumption}
 
\begin{assumption}
\label{assumption:direct2}
The equilibrium point of \eqref{eq:eq_dpi} is asymptotically stable.
\end{assumption}

 Assumption \ref{assumption:direct1} implies that the unforced equilibrium lies in a regime where the local Lipschitz constants of $\mathbf{f}(\mathbf{Z}\mathbf{s}^{*} + \mathbf{\bar{v}})$ violate the conditions of Proposition \ref{prop:direct}.  In light of Assumption \ref{assumption:direct2}, the term $c\mathbf{Bu}^{*}$ can be interpreted as additional active/reactive power injection/consumption that moves the system voltages to a regime where the local Lipschitz constants meet the requirements of Proposition \ref{prop:direct}.

Let $\bm{\mu} = \mathbf{u} - \mathbf{u}^{*}$.  The system of \eqref{eq:sdot_dpi} (which we will now refer to as the \emph{plant}) can then be expressed as:
\begin{gather}
    \mathbf{T}\mathbf{\dot{s}} = \mathbf{f}(\mathbf{Z}\mathbf{s} + \mathbf{\bar{v}} + c\mathbf{B}\bm{\mu} + c\mathbf{B}\mathbf{u}^{*}) -\mathbf{s} , \label{eq:sdot_dpi_mu}
\end{gather}

In order to derive an adaptive controller that will stabilize the system one may compare the unstable plant to a stable reference model.  An adaptation law for $\bm{\mu}$ that will stabilize the plant can then be determined via minimizing the error between the plant and the reference system.  For the system of \eqref{eq:sdot_dpi_mu}, consider the following reference model:
\begin{gather}
    \mathbf{T}\mathbf{\dot{s}}_{r} = \mathbf{f}(\mathbf{Z}\mathbf{s}_{r} + \mathbf{\bar{v}} + c\mathbf{B}\mathbf{u}^{*}) -\mathbf{s}_{r}. \label{eq:sdot_dpi_mu_ref}
\end{gather}

Note that the plant and reference models have the same equilibria (as $\bm{\mu}^{*} = \mathbf{0}$).  Define $\bm{\alpha} = \mathbf{Zs}_{r} + \bar{\mathbf{v}} + c\mathbf{Bu}^{*}$ and let $\mathbf{e}_{s} = \mathbf{s} - \mathbf{s}_{r}$ denote the error between the power injected by the plant and the reference models.  The error dynamics can then be expressed as:
\begin{equation}
    \mathbf{T}\dot{\mathbf{e}}_{s} = \mathbf{f}(\mathbf{Ze}_{s} + c\mathbf{B}\bm{\mu} + \bm{\alpha}) - \mathbf{f}(\bm{\alpha}) - \mathbf{e}_{s} \label{eq:es_dpi}.
\end{equation}

The following theorem establishes the adaptation law that will drive $\bm{\mu}$ to $0$, hence stabilizing \eqref{eq:es_dpi}.
\begin{theorem}
\label{thm:direct_power_inj}
Given the system of \eqref{eq:sdot_dpi_mu}, the reference model \eqref{eq:sdot_dpi_mu_ref}, and the associated error system of \eqref{eq:es_dpi}, suppose Assumptions \ref{assumption:direct1} and \ref{assumption:direct2} hold.  Additionally, suppose that there exists symmetric positive definite matrices $\mathbf{P} \in \mathcal{R}^{n\times n}$, $\mathbf{H} \in \mathcal{R}^{(l+m) \times (l+m)}$, $\mathbf{\Gamma}_{p} \in \mathcal{R}^{l \times l}$, and $\mathbf{\Gamma}_{q} \in \mathcal{R}^{m \times m}$. Define $\mathbf{\Gamma} =\text{diag}\qty(\mathbf{\Gamma}_{p},
\mathbf{\Gamma}_{q})$, $\hat{\mathbf{I}} = [\mathbf{I}_{n \times n}, \mathbf{I}_{n \times n}]^\top$ and the matrix:
\begin{gather}
    \mathbf{\Lambda} = \mathbf{PZT}^{-1}\mathbf{C_{s}}\mathbf{B} - \hat{\mathbf{I}}^\top\mathbf{B\Gamma H} \label{eq:Lambda}.
\end{gather}

Suppose now that $\mathbf{\Lambda}$ is non-positive (i.e., all elements are less than or equal to 0).  Finally, define $\mathbf{e}_{v} = \mathbf{v} - \mathbf{v}_{r}$, where $\mathbf{v}$ is defined in \eqref{eq:v}, and $\mathbf{v}_{r} = \mathbf{Zs}_{r} + \mathbf{\bar{v}}$.  Then the adaptation law:
\begin{equation}
    \dot{\bm{\mu}} = -c\mathbf{\Gamma}\mathbf{B}^\top\hat{\mathbf{I}}\begin{vmatrix}\mathbf{e}_{v}\end{vmatrix} \label{eq:mu_dot},
\end{equation}

\noindent where the absolute value is applied element-wise, asymptotically stabilizes $\mathbf{e}_{s}$.
\end{theorem}

\begin{proof}
The result can be proven via Lyapunov analysis.  Clearly, $\mathbf{e}_{v} = \mathbf{Z}\mathbf{e}_{s}$ and $\dot{\mathbf{e}}_{v} = \mathbf{Z}\dot{\mathbf{e}}_{s}$. Consider the Lyapunov function:
\begin{equation}
    V = \frac{1}{2}\Big( \qty(\mathbf{Z}\mathbf{e}_{s})^\top\mathbf{P}\mathbf{Z}\mathbf{e}_{s} + \bm{\mu}^\top\mathbf{H}\bm{\mu}\Big) \label{eq:Ve}.
\end{equation}

Let $\mathbf{M} = \mathbf{Z}^\top\mathbf{PZ}$. The time derivative of \eqref{eq:Ve} can be expressed as:
\begin{subequations}
\begin{align}
    \dot{V} &= \mathbf{e}_{s}^\top\mathbf{MT}^{-1}\Big(\mathbf{f}(\mathbf{Z}\mathbf{e}_{s} + c\mathbf{B}\bm{\mu} + \bm{\alpha}) - \mathbf{f}( \bm{\alpha})\Big) \nonumber \\
        &- \frac{1}{2}\mathbf{e}_{s}^\top\Big(\mathbf{T}^{-1}\mathbf{M} + \mathbf{MT}^{-1} \Big) \mathbf{e}_{s}
         - c\bm{\mu}^\top\mathbf{H}\mathbf{\Gamma}\mathbf{B}^\top\hat{\mathbf{I}}\begin{vmatrix}\mathbf{e}_{v}\end{vmatrix} \label{eq:Ve_dot}\\
    &\le \begin{vmatrix}\mathbf{Ze}_{s}\end{vmatrix}^\top\mathbf{PZT}^{-1}\begin{vmatrix}\mathbf{f}(\mathbf{Z}\mathbf{e}_{s} + c\mathbf{B}\bm{\mu} + \mathbf{\alpha}) - \mathbf{f}( \mathbf{\alpha})\end{vmatrix} \nonumber \\
    &- \frac{1}{2}\mathbf{e}_{s}^\top\Big(\mathbf{T}^{-1}\mathbf{M} + \mathbf{MT}^{-1} \Big) \mathbf{e}_{s} - c\bm{\mu}^\top\mathbf{H}\mathbf{\Gamma}\mathbf{B}^\top\begin{vmatrix}\mathbf{e}_{v}\end{vmatrix} \label{eq:Ve_dot2}\\
    &\le \begin{vmatrix}\mathbf{Ze}_{s}\end{vmatrix}^\top\mathbf{PZT}^{-1}\begin{vmatrix}\mathbf{C}_{s}\mathbf{Ze}_{s} + c\mathbf{C}_{s}\mathbf{B}\bm{\mu}\end{vmatrix} \nonumber \\
    &- \frac{1}{2}\mathbf{e}_{s}^\top\Big(\mathbf{T}^{-1}\mathbf{M} + \mathbf{MT}^{-1} \Big) \mathbf{e}_{s} - c\bm{\mu}^\top\mathbf{H}\mathbf{\Gamma}\mathbf{B}^\top\begin{vmatrix}\mathbf{e}_{v}\end{vmatrix} \label{eq:Ve_dot3}.
\end{align}
\end{subequations}

\noindent where the absolute value is applied element wise in \eqref{eq:Ve_dot} - \eqref{eq:Ve_dot3}.  Recall that the scalar parameter $c$ encodes the decision to either inject power ($c > 0$) or consume power ($c < 0$).  The nature of $c$ coupled with the definition of $\bm{\mu}$ implies that the product $c\bm{\mu}$ is always vector consisting of positive elements (i.e., if $c>0 \implies \bm{\mu} > 0$ and $c<0 \implies \bm{\mu} < 0$).  This, in conjunction with the conditions surrounding \eqref{eq:Lambda}, imply:
\begin{subequations}
\begin{align}
    &\begin{vmatrix}\mathbf{Ze}_{s}\end{vmatrix}^\top\mathbf{PZT}^{-1}\begin{vmatrix}c\mathbf{C}_{s}\mathbf{B}\bm{\mu}\end{vmatrix} - c\bm{\mu}^\top\mathbf{H}\mathbf{\Gamma}\mathbf{B}^\top\begin{vmatrix}\mathbf{e}_{v}\end{vmatrix} = \label{eq:Ve_dot_mu} \\
    &~~~~~~\begin{vmatrix}\mathbf{e}_{v}\end{vmatrix}^\top\mathbf{PZT}^{-1}c\mathbf{C}_{s}\mathbf{B}\bm{\mu} - c\bm{B}^\top\mathbf{H}\mathbf{\Gamma}\mathbf{B}^\top\begin{vmatrix}\mathbf{e}_{v}\end{vmatrix} \le 0, \label{eq:Ve_dot_mu2}
\end{align}
\end{subequations}

\noindent where \eqref{eq:Ve_dot_mu2} follows from the fact that $\mathbf{e}_{v} = \mathbf{Ze}_{s}$ and $c\mathbf{C}_{s}\mathbf{B}\bm{\mu}$ is a vector with non-negative elements.  Eqs. \eqref{eq:Ve_dot_mu} - \eqref{eq:Ve_dot_mu2} allow $\dot{V}$ to be further simplified:
\begin{subequations}
\begin{align}
    \dot{V} &\le \begin{vmatrix}\mathbf{Ze}_{s}\end{vmatrix}^\top\mathbf{PZT}^{-1}\begin{vmatrix}\mathbf{C}_{s}\mathbf{Ze}_{s} \end{vmatrix} \nonumber \\
    &- \frac{1}{2}\mathbf{e}_{s}^\top\Big(\mathbf{T}^{-1}\mathbf{M} + \mathbf{MT}^{-1} \Big) \mathbf{e}_{s} \label{eq:Ve_dot4} \\
    &\le \begin{Vmatrix}
            \mathbf{e}_{s}
        \end{Vmatrix}_{2} \cdot \begin{Vmatrix}
            \mathbf{MT}^{-1}
            \end{Vmatrix}_{2} \cdot \begin{Vmatrix}
            \mathbf{C}_{s}\mathbf{Z}\mathbf{e}_{s}
        \end{Vmatrix}_{2} \nonumber \\
        & -\begin{Vmatrix}
            \mathbf{MT}^{-1}
            \end{Vmatrix}_{2} \cdot \begin{Vmatrix}
            \mathbf{e}_{s}
        \end{Vmatrix}_{2}^{2}. \label{eq:Ve_dot5}
\end{align}
\end{subequations}

At this stage, note that \eqref{eq:Ve_dot5} is equivalent to an upper bound on the the derivative of a Lyapunov function for an unforced version of \eqref{eq:es_dpi} (the bound being \eqref{eq:app_ve_dot3}), which is described in Appendix  \ref{subsec:appendix_proof}, where $\mathbf{M} = \mathbf{Z}^\top\mathbf{PZ}$.  Proposition \ref{prop:error} in Appendix \ref{subsec:appendix_proof} implies \eqref{eq:Ve_dot5} can be further simplified to $\dot{V} \le -\sigma \norm{\mathbf{e}_{s}}_{2}^{2} \le \mathbf{0}$, where $\sigma$ is a positive constant. 

 La Salle's Invariance Principle \cite{Khalil:1173048} can now be applied to show asymptotic stability of $\bm{\mu}$. 
 As the equilibrium of \eqref{eq:es_dpi} and \eqref{eq:mu_dot} is the origin, it is straightforward to verify that trajectories confined to the set $S = \{\mathbf{e}_{s} : \norm{\mathbf{e}_{s}}_{2}^{2} = \mathbf{0}\}$ imply $\bm{\mu} = \mathbf{0}$.  Therefore, it can be concluded that $\bm{\mu}$ is asymptotically stable. 
\end{proof}

Asymptotic stability of \eqref{eq:es_dpi} implies that $\mathbf{s} \rightarrow \mathbf{s}_{r}$ as $t \rightarrow \infty$.  The adaptation law of \eqref{eq:mu_dot} implies that $\mathbf{u}$ will stop changing when the voltages observed by the plant and the reference are equivalent.  This will occur once $\mathbf{u}$ has moved the measured voltage of the plant into a regime where the unstable VV/VW curves are generating constant power with respect to changing voltages.  

As a consequence of \eqref{eq:Lambda}, satisfaction of Theorem \ref{thm:direct_power_inj} necessitates the construction of matrices $\mathbf{P}$, $\mathbf{H}$, $\mathbf{\Gamma}_{p}$, and $\mathbf{\Gamma}_{q}$ to ensure $\mathbf{\Lambda}$ is non-positive.  Choices of $\mathbf{\Gamma}_{p}$ and $\mathbf{\Gamma}_{q}$ have important implications for the deployment of the scheme (e.g., whether the approach can be implemented locally and without communication between controllers).  To obtain a local and communication-free implementation, it is necessary to choose $\mathbf{\Gamma}_{p}$ and $\mathbf{\Gamma}_{q}$ to be diagonal.  However, as \eqref{eq:Lambda} is a function of the network topology and impedances ($\mathbf{Z}$), the Lipshitz constants of VV/VW controllers $\mathbf{C}_{s}$, and the nodal location of controllable direct power injection devices ($\mathbf{B}$), certain deployments of controllers in some networks may prevent construction of a non-positive $\mathbf{\Lambda}$.  In the situation where controllable direct power injection devices exist at every node in $\mathcal{G}$ (in this case $\mathbf{B} = \mathbf{I}_{2n \times 2n}$) then non-positivity of $\mathbf{\Lambda}$ can be guaranteed as $\mathbf{H}$ can be chosen to be non-negative, symmetric, positive definite, and sufficiently large.

\subsection{Adaptive Voltage Bias}
\label{subsec:bias}

In this section, we consider an alternative strategy to adaptive power injection/consumption whereby the system of \eqref{eq:sdot} - \eqref{eq:v} is stabilized via introducing a bias term into the measured voltage to \textquote{trick} non-compromised smart inverters into operating in a stable region where the Lipschitz constants of the VV and VW controllers meet the conditions of Proposition \ref{prop:direct}.  Biasing the voltage signal input to the VV/VW controllers is accomplished via adding a term to the voltage measured directly from the grid.  This strategy is equivalent to translating the VV/VW curves along their voltage axis, similar to \cite{singhal2019real}.  As the adaptation law only adjusts the input to the VV/VW control functions, the strategy is compliant with emerging DER standards for PV systems \cite{IEEE_1547, rule21, inverter2016}.  

To begin, consider $l$ photovoltaic smart inverters equipped with adaptive voltage bias control located at a subset of nodes in $\mathcal{G}$ (e.g., $l \le n$).  Let $\mathbf{w} \in \mathcal{R}^{l\times 1}$ denote the collection of voltage offsets, or biases, that will be added to the voltage measurements input into VV/VW functions associated with these controllers.  The sparse matrix $\mathbf{D} \in \mathcal{R}^{n\times l}$ will determine which bias terms in $\mathbf{w}$ will be added to inputs of VV/VW controllers at nodal locations in $\mathcal{G}$.  $\mathbf{D}(i,j) = 1$ indicates that the adaptive bias at entry $j$ in $\mathbf{w}$ will be added to VV/VW functions at node $i$ in $\mathcal{G}$.  As such each row and column of $\mathbf{D}$ consists of all 0s with at most a single entry equal to 1. With these definitions, the effect of voltage biasing on the dynamics of \eqref{eq:sdot} - \eqref{eq:v} can be expressed as:
\begin{gather}
    \mathbf{T}\mathbf{\dot{s}} = \mathbf{f}(\mathbf{Z}\mathbf{s} + \mathbf{\bar{v}}  +d\mathbf{D}\mathbf{w}) -\mathbf{s}, \label{eq:sdot_bias}
\end{gather}

\noindent which has equilibrium:
\begin{equation}
        \mathbf{0} = \mathbf{f}(\mathbf{Z}\mathbf{s}^{*} + \mathbf{\bar{v}} + d\mathbf{D}\mathbf{w}^{*}) -\mathbf{s}^{*} \label{eq:eq_bias}.
\end{equation}

Here the parameter $d \in \{-1,1\}$ indicates if the bias will be positive ($d=1$) or negative ($d=-1$) and $\mathbf{w} \in \mathbb{R}^{n \times 1}$ (without loss of generality) is the vector of bias magnitudes.  Note that this formulation implicitly assumes that all devices participating in this control activity will bias the measured voltages in the same direction.  The following assumptions are made regarding the forced and unforced equilibria of \eqref{eq:sdot_dpi}:

\begin{assumption}
\label{assumption:bias1}
The equilibrium of the unforced system of \eqref{eq:sdot_bias} (Eq.\eqref{eq:eq_bias} where $\mathbf{w}^{*} = \mathbf{0}$) is unstable.
\end{assumption}
 
\begin{assumption}
\label{assumption:bias2}
The equilibrium point of \eqref{eq:eq_bias} is asymptotically stable.
\end{assumption}

Assumption \ref{assumption:bias1} implies that the unforced equilibrium lies in a regime where the local Lipschitz constants of $\mathbf{f}(\mathbf{Z}\mathbf{s}^{*} + \mathbf{\bar{v}})$ violate the conditions of Proposition \ref{prop:direct}.  In light of Assumption \ref{assumption:bias2}, the term $d\mathbf{Du}^{*}$ can be interpreted as a translation of the VV/VW curves along the voltage axis that provide additional active/reactive power injection/consumption that moves the system voltages to a regime where the local Lipschitz constants meet the requirements of Proposition \ref{prop:direct}.  

Let $\bm{\phi} = \mathbf{w} - \mathbf{w}^{*}$. The system of \eqref{eq:sdot_bias} (henceforth refer to as the \emph{plant}) can then be expressed as:
\begin{gather}
    \mathbf{T}\mathbf{\dot{s}} = \mathbf{f}(\mathbf{Z}\mathbf{s} + \mathbf{\bar{v}}  + d\mathbf{D}\bm{\phi} + d\mathbf{D}\mathbf{w}^{*}) -\mathbf{s}. \label{eq:sdot_bias_phi}
\end{gather}

For this system, consider the following reference model:
\begin{gather}
    \mathbf{T}\mathbf{\dot{s}}_{r} = \mathbf{f}(\mathbf{Z}\mathbf{s}_{r} + \mathbf{\bar{v}} + d\mathbf{D}\mathbf{w}^{*}) -\mathbf{s}_{r}. \label{eq:sdot_bias_phi_ref}
\end{gather}

Note that the plant and reference models have the same equilibria (as $\bm{\phi}^{*} = \mathbf{0}$).  Define $\bm{\alpha} = \mathbf{Zs}_{r} + \bar{\mathbf{v}} + d\mathbf{Du}^{*}$ and let $\mathbf{e}_{s} = \mathbf{s} - \mathbf{s}_{r}$ denote the error between the power injected by the plant and the reference models.  The error dynamics can then be expressed as:
\begin{equation}
    \mathbf{T}\dot{\mathbf{e}}_{s} = \mathbf{f}(\mathbf{Ze}_{s} + d\mathbf{D}\bm{\phi} + \bm{\alpha}) - \mathbf{f}(\bm{\alpha}) - \mathbf{e}_{s} \label{eq:es_bias}.
\end{equation}

The following theorem establishes the adaptation law that will drive $\bm{\phi}$ to $0$, hence stabilizing \eqref{eq:es_bias}.
\begin{theorem}
\label{thm:bias}
Given the system of \eqref{eq:sdot_bias_phi}, the reference model \eqref{eq:sdot_bias_phi_ref}, and the associated error system of \eqref{eq:es_bias}, suppose Assumptions \ref{assumption:bias1} and \ref{assumption:bias2} hold.  Additionally, suppose that there exists symmetric positive definite matrices $\mathbf{P} \in \mathcal{R}^{n\times n}$, $\mathbf{H} \in \mathcal{R}^{l \times l}$, and $\mathbf{\Gamma}_{v} \in \mathcal{R}^{l \times l}$.  Define the matrix:
\begin{gather}
    \mathbf{\Theta} = \mathbf{PZT}^{-1}\mathbf{C}_{s}\mathbf{D} - \mathbf{D\Gamma}_{v} \mathbf{H} \label{eq:Theta}.
\end{gather}

Suppose now that $\mathbf{\Theta}$ is non-positive (i.e., all elements are less than or equal to 0).  Finally, define $\mathbf{e}_{v} = \mathbf{v} - \mathbf{v}_{r}$, where $\mathbf{v}$ is defined in \eqref{eq:v}, and $\mathbf{v}_{r} = \mathbf{Zs}_{r} + \mathbf{\bar{v}}$.  Then the adaptation law:
\begin{equation}
    \dot{\bm{\phi}} = -d\mathbf{\Gamma}_{v}\mathbf{D}^\top\begin{vmatrix}\mathbf{e}_{v}\end{vmatrix} \label{eq:phi_dot},
\end{equation}

\noindent where the absolute value is applied element-wise, asymptotically stabilizes $\mathbf{e}_{s}$.
\end{theorem}

\begin{proof}
The result can be proven via Lyapunov analysis.  Note that, $\mathbf{e}_{v} = \mathbf{Z}\mathbf{e}_{s}$ and $\dot{\mathbf{e}}_{v} = \mathbf{Z}\dot{\mathbf{e}}_{s}$. Consider the Lyapunov function:
\begin{equation}
    V = \frac{1}{2}\Big( \qty(\mathbf{Z}\mathbf{e}_{s})^\top\mathbf{P}\mathbf{Z}\mathbf{e}_{s} + \bm{\phi}^\top\mathbf{H}\bm{\phi}\Big) \label{eq:Ve_bias}.
\end{equation}

\bhl{Noting the similar structure of \eqref{eq:es_bias} and \eqref{eq:es_dpi}, the remainder of the proof proceeds identically to that of Theorem \ref{thm:direct_power_inj} (where $c\mathbf{B}\bm{\mu}$ is replaced with $d\mathbf{D}\bm{\phi}$).}
\end{proof}

Asymptotic stability of \eqref{eq:es_bias} implies that $\mathbf{s} \rightarrow \mathbf{s}_{r}$ as $t \rightarrow \infty$.  The adaptation law of \eqref{eq:phi_dot} implies that $\mathbf{w}$ will stop changing when the voltages observed by the plant and the reference are equivalent.  This will occur once $\mathbf{w}$ has moved the measured voltage of the plant into a regime where the unstable VV/VW curves are generating constant power with respect to changing voltages.  

Similarly to \eqref{eq:Lambda}, as a consequence of \eqref{eq:Theta}, satisfaction of Theorem \ref{thm:bias} necessitates the construction of matrices $\mathbf{P}$, $\mathbf{H}$, and $\mathbf{\Gamma}_{v}$ to ensure $\mathbf{\Theta}$ is non-positive.  The choice of $\mathbf{\Gamma}_{v}$ again has important implications for the deployment of the scheme (e.g., whether the approach can be implemented locally and without communication between controllers).  To obtain a local and communication-free implementation, it is necessary to choose $\mathbf{\Gamma}_{v}$ to be diagonal.  However, as \eqref{eq:Theta} is a function of the network topology and impedances ($\mathbf{Z}$), the Lipshitz constants of VV/VW controllers $\mathbf{C}_{s}$, and the nodal location of VV/VW controllers with adaptive voltage bias control ($\mathbf{D}$), certain deployments of controllers in some networks may prevent construction of a non-positive $\mathbf{\Theta}$.  In the situation where inverter resources at all nodes in $\mathcal{G}$ are executing voltage bias control (in this case $\mathbf{D} = \mathbf{I}_{n \times n}$) then non-positivity of $\mathbf{\Theta}$ can be guaranteed as $\mathbf{H}$ can be chosen to be non-negative, symmetric, positive definite, and sufficiently large.
\subsection{Implementation}
\label{subsec:implementation}

As previously mentioned in Sections \ref{subsec:direct_power_injection} - \ref{subsec:bias}, if $\mathbf{\Gamma}_{p}$, $\mathbf{\Gamma}_{q}$, and $\mathbf{\Gamma}_{v}$ are chosen to be positive definite and diagonal then the adaptation laws of \eqref{eq:mu_dot} and \eqref{eq:phi_dot} admit a local implementation, where individual elements of $\mathbf{u}_{p}$, $\mathbf{u}_{q}$, and $\mathbf{w}$ rely only on a single element of $\mathbf{e}_{v}$ which corresponds to the node at which the adaptive power injection or adaptive bias controller is located within $\mathcal{G}$.  In other words $u_{p,i} = -d\gamma_{p,i}\vert e_{v,i} \vert$, for example.  

While the implementation is local and communications free, the individual adaptive power injection and adaptive bias controllers still require the signal $e_{v,i}$, which is the error between node $i$ voltage magnitude and a voltage reference, given by $v_{r,i}$.  This reference voltage is fictitious and should capture the behavior of the voltage at node $i$ in the absence of smart inverter-driven instabilities.  As these instabilities manifest as large oscillations, an extremely effective proxy for the reference voltage is the low pass filtered node $i$ voltage magnitude.  Proper low pass filtering of $v_{i}$ can remove all traces of the instability resulting in a stable signal.  Noting that $\mathbf{v}$ and the signal produced by low-pass filtering of $\mathbf{v}$ have the same equilibrium values, the error signal $\mathbf{e}_{v}$ will become approximately zero when $\mathbf{v}$ stabilizes.

Let $\xi_{i}$ denote the low pass filtered voltage magnitude at node $i$.  Control laws for both adaptive power injection (active and reactive) and adaptive voltage bias for an arbitrary node $i \in \mathcal{G}$ are:
\begin{subequations}
    \begin{align}
        \dot{\xi}_{i} &= \tau_i \big( v_{i} - \xi_{i} \big) \label{eq:xi_i}\\
        \dot{u}_{p,i} &= -c_{i}\gamma_{p,i}\lvert v_{i} - \xi_{i} \rvert \mathbf{1}^{\epsilon}_{i}\label{eq:upi}\\
        \dot{u}_{q,i} &= -c_{i}\gamma_{q,i}\lvert v_{i} - \xi_{i} \rvert \mathbf{1}^{\epsilon}_{i}\label{eq:uqi}\\
        \dot{w}_{i} &= -d_{i}\gamma_{v,i}\lvert v_{i} - \xi_{i} \rvert \mathbf{1}^{\epsilon}_{i}\label{eq:wi}\\
        \mathbf{1}^{\epsilon}_{i} &= \begin{cases}
                                1 \quad \text{if} \quad \lvert v_{i} - \xi_{i} \rvert > \epsilon_{i}\\
                                0 \quad \text{else}
                            \end{cases}
\end{align}
\end{subequations}

\noindent where $\mathbf{1}^{\epsilon}_{i}$ is an indicator function that will prevent adaptation for small differences between $v_{i}$ and $\xi_{i}$.  The parameters $\tau_i$, $\gamma_{p,i}$, $\gamma_{q,i}$, and $\gamma_{v,i}$ are positive scalars.  Note that the adaptation laws for adaptive power injection for both active and reactive power and the adaptive voltage bias differ only by choice of their respective scaling parameters: $\gamma_{p,i}$, $\gamma_{q,i}$, and $\gamma_{v,i}$.

The low pass filter time constant, $\tau_{i}$, can be easily determined in practice using knowledge of the timestep of the smart inverter VV/VW control loops or by observing the frequency of the unstable oscillations in system voltages. As the parameter $c_{i}$ dictates if active/reactive power is to be consumed or injected for adaptive power injection controllers and $d_{i}$ determines the sign of the voltage offset for adaptive voltage bias controllers, we recommend choosing $c$ and $d$ according to a simple heuristic that can be implemented in a distributed fashion.  Choosing $c=-1$ and $d=1$ when the nodal voltage is greater than a pre-defined threshold ($v_{\text{crit}} = 1.0$ p.u. for instance) and $c=1$ and $d=-1$ when the nodal voltage is below the threshold will ensure that the adaptation laws will lower voltage magnitudes if $v_{i} > v_{\text{crit}}$ and will raise voltage magnitudes if $v_{i} < v_{\text{crit}}$.

\bhl{The proposed adaptive control scheme is designed to function as a supervisory controller operating on a slower timescale than other fast acting power electronics-based controllers in the system, e.g. STATCOM and SVR. This ensures that, in the presence of these devices, there are no adverse interactions nor do the proposed controllers impede their operation. Instead, the proposed approach serves as an additional layer of control that would help mitigate oscillatory instabilities in the absence of these devices, or should these devices have insufficient controllabiltiy.}

\section{Simulation Results}
\label{sec:results}
Simulation experiments were conducted on the IEEE $37$ and the IEEE $8500$ test feeders to verify the performance of the adaptive control scheme in mitigating smart inverter-driven voltage instabilities in \emph{three-phase unbalanced} systems.  Let $\mathcal{M}$ denote the set of smart inverters in the feeder.  In these experiments, a subset of smart inverters with VV/VW capability were issued new VV/VW curves steep enough to violate the conditions of Proposition \ref{prop:direct} and create an instability in the form of large oscillations in system voltage magnitudes.  Denote this set as $\mathcal{M}_{u} \subseteq \mathcal{M}$, where the subscript $u$ refers to ``unstable".  The goal of the adaptive bias control was to dynamically adjust the voltage input into the VV/VW functions of inverters in the set $\mathcal{M}_{s} \subseteq \mathcal{M}$.  Adaptive power injection control, instead, did not utilize inverters in $\mathcal{M}_{s}$, but rather increased reactive power consumption at nodes in $\mathcal{M}_{s}$ to stabilize system voltages.

Both the adaptive voltage bias and adaptive power injection control strategies were separately tested in two different scenarios.  \bhl{In \textbf{Scenario 1} compromised smart inverters and non-compromised smart inverters with adaptive controllers are co-located at the same nodes in the feeder, or $\mathcal{M}_{s} = \mathcal{M}_{u}$, (see Fig. \ref{fig:ieee37_uniform}).  In \textbf{Scenario 2} destabilizing smart inverters and adaptive controllers are placed at different nodes in the feeder, or $\mathcal{M}_{s} \ne \mathcal{M}_{u}$ (see Fig. \ref{fig:ieee37_diff}). Both figures depict the fraction of smart inverters responsible for creating instabilities in red (i.e., $\mathcal{M}_{u}$) and the remaining portion of smart inverters which can be utilized for adaptive voltage bias control in green (i.e., $\mathcal{M}_{s}$).}  In all experiments, smart inverters had a peak active power generation of 100\% of the nominal load with an additional 10\% inverter over-sizing for reactive power headroom. \bhl{This oversizing is required to meet anticipated reactive power capabilities at maximum active power output and is consistent with a recent study demonstrating the benefits for feeder hosting capacity\cite{nyserda}. An alternative implementation to provide this reactive capability is by curtailing active power generation when necessary to reserve headroom. The analysis presented in this work is equally valid for both implementations.} Simulations were conducted in OpenDSS with a timestep of 1 second. \bhl{As described in Section \ref{sec:control_design}, the proposed controller is intended to operate on a slower timescale with respect to other fact-acting power electronic devices that may be on the network. This assumption allows the use of a quasi steady-state (QSS) approximate model when carrying out simulations.}

\bhl{
We utilize an intuitive filtering process to extract the \textquote{energy} associated with observed voltage oscillations.  The filter consists of the series connection of a high-pass filter $H_{HP}$, a signal square element (with positive gain $c$), and a low-pass filter $H_{LP}$, shown in Fig. \ref{fig:observer}.  The output of the filter $y_{i}$ is a non-negative value which becomes larger as the amplitude of the oscillations in node $i$ voltage ($v_{i}$) increase.  For proper operation, the high and low-pass filter critical frequencies should be chosen as to not attenuate oscillations resulting from cyber-attacked inverters.
\begin{figure}[ht]
\centering
\includegraphics[width=1\columnwidth]{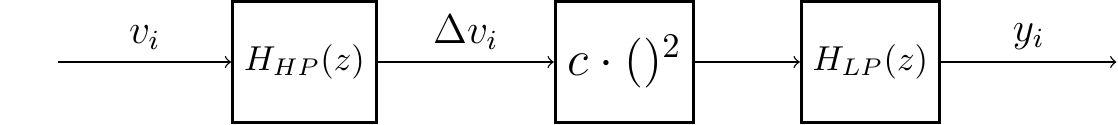}
\caption{Block diagram of illustrating the filtering process used to compute a measurement of the intensity of voltage oscillations.}
\label{fig:observer}
\end{figure}
}

\bhl{
We now present simulation results highlighting the performance of the adaptive control scheme on the IEEE 37 node feeder (Scenario 1 and Scenario 2) and the IEEE 8500 test feeder (Scenario 1). Note that we do not explicitly compare the performance of our controller to the approach undertaken by Singhal et. al. \cite{singhal2019real} as in our simulations we assume that a subset of PV devices are solely responsible for creating instabilities and are not controllable.  This renders the approach approach in \cite{singhal2019real} ineffective in mitigating inverter-driven oscillations.

Simulation results for all feeders/scenarios are depicted in Figs. \ref{fig:exp_result} - \ref{fig:exp_result_y_stats}, where distinct experiments are ordered column-wise.  In the figures, the first column shows results for Scenario 1 on the IEEE 37 node test feeder, column two shows results for Scenario 2 on the IEEE 37 node test feeder, and column three shows results for Scenario 1 on the IEEE 8500 test feeder.  Additionally, each column of Figure \ref{fig:exp_result} depicts three voltage timeseries from a select node and phase in each respective feeder featuring a base case without adaptive control (show in blue), the results from the application of adaptive voltage bias (show in purple), and results from the application of adaptive power injection (show in green).  Directly under each voltage timeseries subplot is an additional timeseries showing the associated oscillation \textquote{energy} (i.e., $y$ from Fig. \ref{fig:observer}).  Figure \ref{fig:exp_result_y_stats} depicts feeder-wide statistics of $y$ as a function of simulation time for each separate feeder/scenario, for the base case, adaptive voltage bias, and adaptive power injection.  We now discuss simulation results associated with each feeder/scenario.  Note that parameters of the adaptive controllers used in these experiments can be found in Table \ref{table:params}.
}

\begin{description}[leftmargin=0pt]
\bhl{
\item[Scenario 1 - IEEE 37 Node Feeder:] Simulation results for the smart inverter deployment depicted in Fig. \ref{fig:ieee37_uniform} are shown in Fig. \ref{fig:exp_result} column 1.  In these experiments, at t = 100s, inverters in the set $\mathcal{M}_{u}$, which represent 30\% of the inverter resource at each node, were issued new VV/VW curves with steeper non-zero segments, resulting in an instability.  The voltage magnitude at node $741$ (phase C) without the presence of any adaptive control and the associated oscillation energy are shown subplots 1-2 of Fig. \ref{fig:exp_result} column 1.  The subplots 3-4 depict node 741 phase C voltage magnitude when inverters in $\mathcal{M}_{s}$ utilize adaptive voltage bias and the associated oscillation energy.  Subplots 5-6 depict node 741 phase C voltage when adaptive (reactive) power injection is employed at nodes where inverters in $\mathcal{M}_{s}$ are located and the associated oscillation energy.  Both control strategies mitigate the oscillations within approximately 80 seconds of the onset of the instability.  The results in Fig. \ref{fig:exp_result_y_stats}, column 1, indicate that both the adaptive power injection and the adaptive voltage bias controllers mitigate oscillations at all nodes in the system.\\
 }
\begin{figure}[h!]
\centering
\includegraphics[width=1\columnwidth]{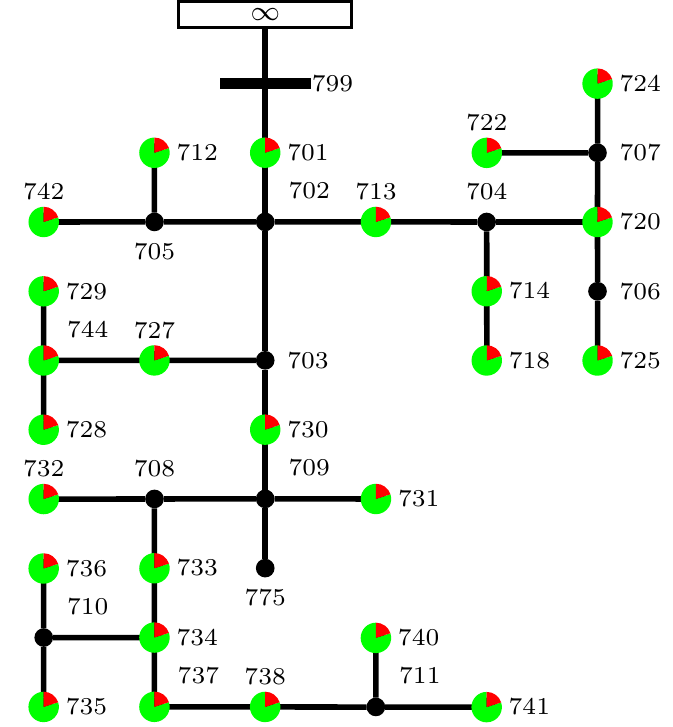}
\caption{IEEE 37 node test feeder with smart inverters from the sets $\mathcal{M}_{s}$ (with normalized capacity represented by green) and $\mathcal{M}_{u}$ (with normalized capacity represented by red) co-located at the same node (Scenario 1).}
\label{fig:ieee37_uniform}
\end{figure}
 
 \begin{figure}[h!]
\centering
\includegraphics[width=1\columnwidth]{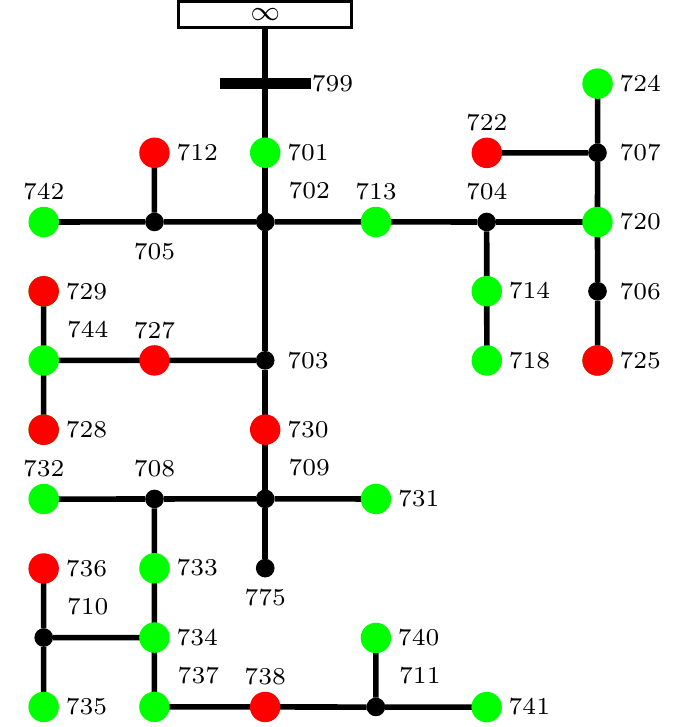}
\caption{IEEE 37 node test feeder with smart inverters from the sets $\mathcal{M}_{s}$ (with normalized capacity represented by green) and $\mathcal{M}_{u}$ (with normalized capacity represented by red) located at different nodes (Scenario 2).}
\label{fig:ieee37_diff}
\end{figure}

\bhl{
\item[Scenario 2 - IEEE 37 Node Feeder:] Simulation results for the smart inverter deployment depicted in Fig. \ref{fig:ieee37_diff} are shown in Fig. \ref{fig:exp_result} column 2. In these experiments, at t = 100s, inverters in the set $\mathcal{M}_{u}$, which represent approximately 30\% of the capacity in the system, were issued new VV/VW curves with steeper non-zero segments, resulting in an instability. The voltage magnitude at node $741$ (phase C) without the presence of any adaptive control and the associated oscillation energy are shown subplots 1-2 of Fig. \ref{fig:exp_result} column 2.  The subplots 3-4 depict node 741 phase C voltage magnitude when inverters in $\mathcal{M}_{s}$ utilize adaptive voltage bias and the associated oscillation energy.  Subplots 5-6 depict node 741 phase C voltage when adaptive (reactive) power injection is employed at nodes where inverters in $\mathcal{M}_{s}$ are located and the associated oscillation energy.  Both control strategies mitigate the oscillations within approximately 100 seconds of the onset of the instability.  The results in Fig. \ref{fig:exp_result_y_stats}, column 2, indicate that both the adaptive power injection and the adaptive voltage bias controllers mitigate oscillations at all nodes in the system.\\
}

\bhl{
\item[Scenario 1 - IEEE 8500 Node Feeder:] Simulation results for the smart inverter deployment on the IEEE 8500 feeder (Scenario 1) are shown in Fig. \ref{fig:exp_result} column 3.  In these experiments, at t = 100s, inverters in the set $\mathcal{M}_{u}$, which represent 30\% of the inverter resource at each node, were issued new VV/VW curves with steeper non-zero segments, resulting in an instability. The voltage magnitude at node $337668b0a$ (phase A) without the presence of any adaptive control and the associated oscillation energy are shown subplots 1-2 of Fig. \ref{fig:exp_result} column 3.  The subplots 3-4 depict node $337668b0a$ phase A voltage magnitude and the associated oscillation intensity when inverters in $\mathcal{M}_{s}$ utilize adaptive voltage bias.  The subplots 5-6 depicts node $337668b0a$ phase A voltage and the associated oscillation intensity when adaptive reactive power injection is employed at nodes where inverters in $\mathcal{M}_{s}$ are located.  Both control strategies mitigate the oscillations within approximately 60 seconds of the onset of the instability. The results in Fig. \ref{fig:exp_result_y_stats}, column 3, indicate that both the adaptive power injection and the adaptive voltage bias controllers mitigate oscillations at all nodes in the system. 
}
\end{description}

\begin{table}[]
\centering
\caption{Adaptive Controller Parameters for both Adaptive Voltage Bias and Direct Power Injection Experiments.}
\label{table:params}
\resizebox{1\columnwidth}{!}{
\begin{tabular}{cccc}
\toprule
  \textbf{Parameter} & \textbf{37 - Scn.1}  & \textbf{37 - Scn.2}  & \textbf{8500} \\
  \cmidrule(lr){1-1} \cmidrule(lr){2-2} \cmidrule(lr){3-3} \cmidrule(lr){4-4}
$\tau_{i}$ (low pass filter time constant) & 0.1 & 0.1 & 0.1\\
$c_{i}$, $d_{i}$ (sign of power injection and bias) & -1, 1 & -1, 1 & -1, 1\\
 $\gamma_{i}$, $\gamma_{q,i}$ (adaptation gains) & 0.1, 20 & 0.2, 50 & 0.2, 1\\
$\epsilon_{i}$ (adaptation threshold) & $10^{-4}$ & $10^{-4}$ & $10^{-4}$\\
\bottomrule
\end{tabular}
}
\end{table}


\begin{figure*}[]
  \centering
  
  
  
\includegraphics[scale=1]{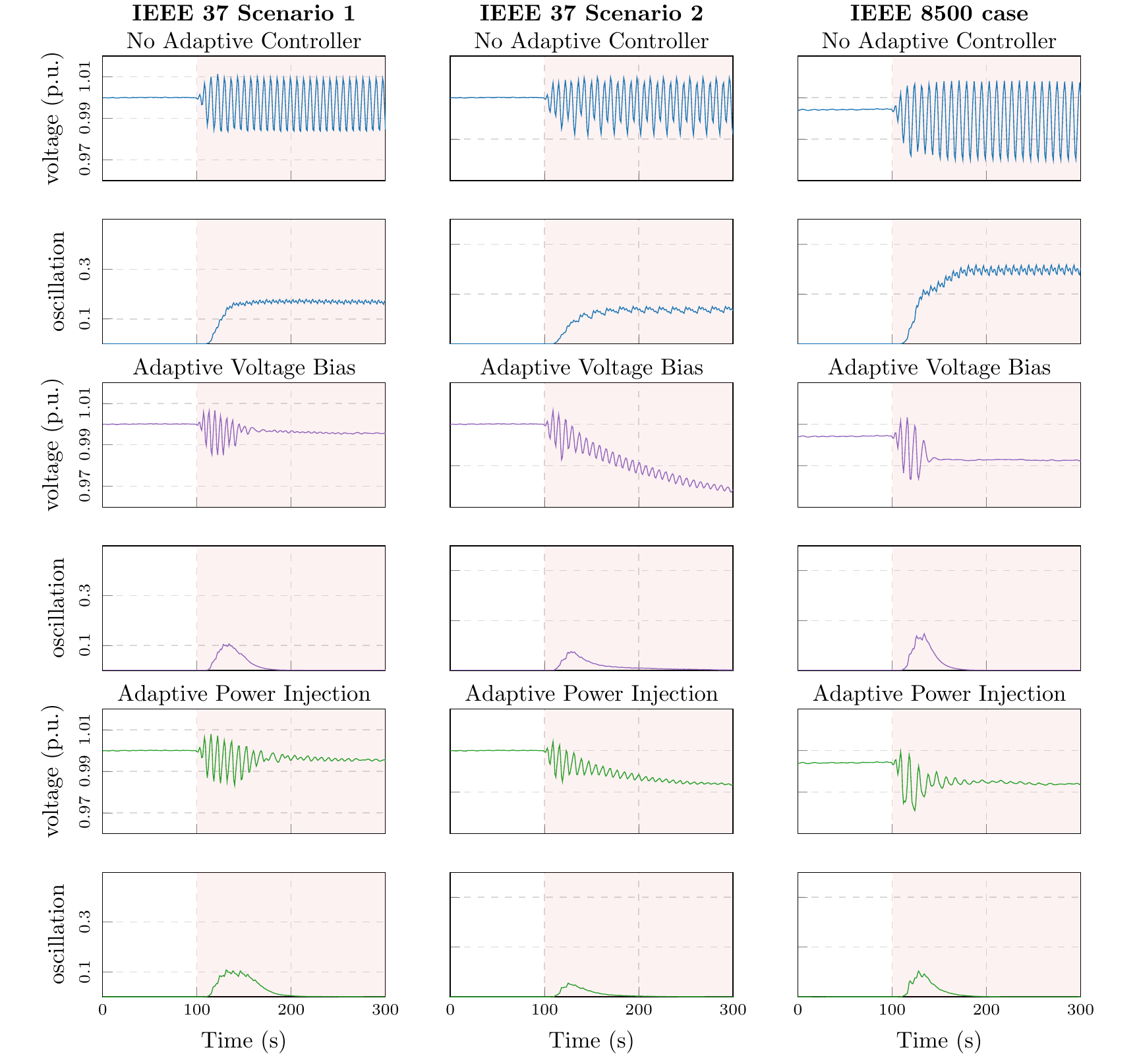}
\caption{\textbf{Column 1:} Node 741 (phase C) voltage during simulation experiments in the scenario depicted in Fig. \ref{fig:ieee37_uniform} (Scenario 1). \textbf{Column 2:} Node 741 (phase C) voltage during simulation experiments in the scenario depicted in Fig. \ref{fig:ieee37_diff} (Scenario 2). \textbf{Column 3:} Node 337668b0a (phase A) voltage during simulation experiments in the IEEE 8500 node feeder in Scenario 1. The shaded red region shows the period of the timeseries where a portion of smart inverters at each node have been issued unstable VV/VW control curves.}
\label{fig:exp_result}
\end{figure*}

\begin{figure*}[t]
  \centering
  
  
  
\includegraphics[scale=1]{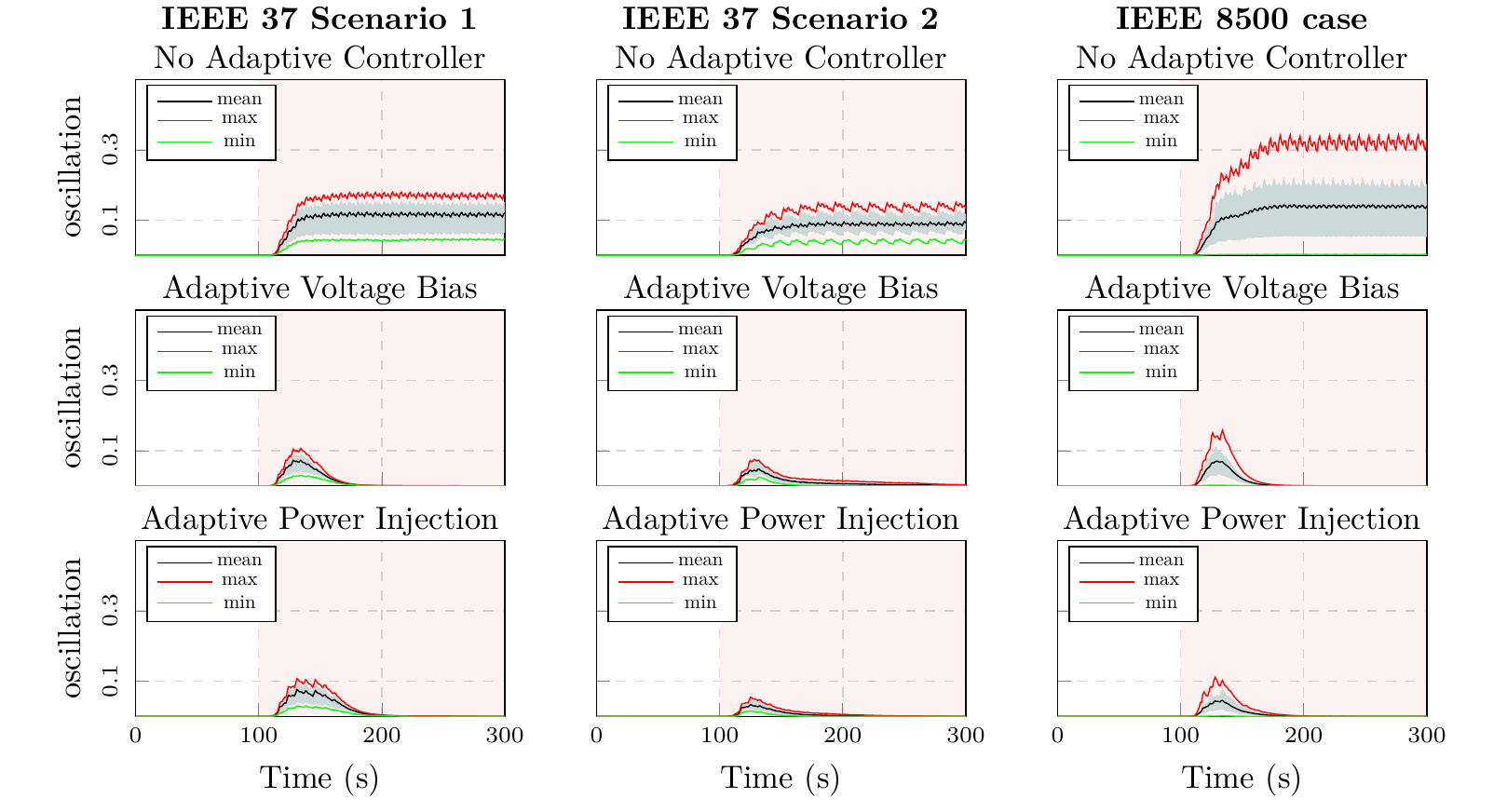}
\caption{\textbf{Column 1:} Oscillation energy during simulation experiments in the scenario depicted in Fig. \ref{fig:ieee37_uniform} (Scenario 1). \textbf{Column 2:} Oscillation energy during simulation experiments in the scenario depicted in Fig. \ref{fig:ieee37_diff} (Scenario 2). \textbf{Column 3:} Oscillation energy during simulation experiments in the IEEE 8500 node feeder in Scenario 1. The shaded blue region shows the 25th and 75th percentile of oscillation energy across all nodes. The shaded red region shows the period of the timeseries where a portion of smart inverters at each node have been issued unstable VV/VW control curves.}
\label{fig:exp_result_y_stats}
\end{figure*}

\section{Discussions}
\label{sec:conclusions}

This paper explored the use of adaptive control to manage a) smart inverters with Volt-VAR and Volt-Watt capabilities and b) devices capable of direct power injection (such as battery storage systems) to mitigate oscillations in the electric grid voltages introduced by portions of smart inverters with unstable VV/VW functions.  Control of smart inverters with Volt-VAR and Volt-Watt control is accomplished via introducing an offset into the voltage magnitude that is input into VV/VW functions.  \bhl{Control of devices capable of direct power injection is accomplished via issuing new active/reactive power injection setpoints.  In both cases, the adaptation law increases the voltage offset or power injection until system voltages are driven into regions where the local Lipschitz constants of destabilizing smart inverters are within stability limits, thus mitigating the oscillations in grid voltages.}  Simulation results show the effectiveness of the approach in managing smart inverters in unbalanced distribution systems. \bhl{We have tested the approach under a variety of different solar profiles and times of day (e.g., morning, noon, mid afternoon) and found no appreciable differences in the performance of the adaptive control scheme. }

The adaptation law is driven by the absolute value of the error between the unstable voltage magnitude and a stable voltage reference signal.  Although this reference signal is fictitious, an extremely effective proxy for the reference is the low pass filtered nodal voltage.  The error between the voltage magnitude and the low pass filtered voltage ``reference" will be approximately 0 once the instabilities are mitigated.  Therefore, the proposed method constitutes an essential mechanism for online, decentralized, and model-free mitigation of smart inverter-driven instabilities immediately after these instabilities manifest in the system.  

\bhl{Given the simplicity of implementation of the proposed adaptive control scheme, in the future we plan to look for opportunities to test the approach in hardware-in-the-loop experiments.}

\appendix
\label{sec:appendix}

\subsection{Smart Inverter Stability Criterion}
\label{subsec:appendix_stability}
\bhl{
The following proposition ties the stability of \eqref{eq:sdot2} to the system impedances, $\mathbf{Z}$, and $\mathbf{C}_{s}$.

\begin{proposition}
\label{prop:direct}
Let $\mathbf{M} \in \mathcal{R}^{2n \times 2n}$ be a positive definite and symmetric matrix.  The system of \eqref{eq:sdot} is asymptotically stable if:
\begin{equation}
    -\underline{\lambda} + \overline{\lambda}  \norm{ \mathbf{C}_{s}\mathbf{Z}}_{2} \le 0 \label{eq:stab_cond},
\end{equation}

\noindent where $\underline{\lambda} = \lambda_{\min}(\mathbf{MT}^{-1})$ and $\overline{\lambda} = \lambda_{\max}(\mathbf{MT}^{-1})$.
\end{proposition}
} 
\begin{proof}
    \bhl{Noting the equilibrium $\mathbf{s}^{*}$ of \eqref{eq:sdot} is 
    \begin{equation}
            \mathbf{0} = \mathbf{f}(\mathbf{Z}\mathbf{s}^{*} + \mathbf{\bar{v}}) -\mathbf{s}^{*}\label{eq:eqpoint},
    \end{equation}
    
    \noindent define the shifted set of coordinates $\Delta \mathbf{s} = \mathbf{s} - \mathbf{s}^{*}$ which translate the equilibrium to the origin.  The dynamics in the new coordinate system are:
    \begin{equation}
        \mathbf{T}\Delta\mathbf{\dot{s}} = \mathbf{f}(\mathbf{Z}\Delta\mathbf{s} + \mathbf{Z}\mathbf{s}^{*} + \mathbf{\bar{v}}) - \mathbf{f}(\mathbf{Z}\mathbf{s}^{*} + \mathbf{\bar{v}}) -\Delta\mathbf{s} \label{eq:delsdot}.
    \end{equation}

Let $\mathbf{\alpha} = \mathbf{Z}\mathbf{s}^{*} + \mathbf{\bar{v}}$.  Using the Lyapunov function $V = \frac{1}{2}\Delta \mathbf{s}^\top\mathbf{M}\Delta \mathbf{s}$, the derivative of the state trajectories of \eqref{eq:delsdot} along $V$ are:
\begin{subequations}
    \begin{align}
        \dot{V} &= \Delta \mathbf{s}^\top \mathbf{MT}^{-1}\Big( \mathbf{f}(\mathbf{Z} \Delta \mathbf{s} + \mathbf{\alpha}) - \mathbf{f}(\mathbf{\alpha}) \Big) \nonumber \\
        & - \frac{1}{2}\Delta \mathbf{s}^\top \big( \mathbf{T}^{-1}\mathbf{M} + \mathbf{MT}^{-1}\big)\Delta \mathbf{s} \label{eq:ve_dot1}\\
        &\le \begin{Vmatrix}
            \Delta \mathbf{s}
        \end{Vmatrix}_{2} \cdot \begin{Vmatrix}
            \mathbf{MT}^{-1}
            \end{Vmatrix}_{2} \cdot \begin{Vmatrix}
            \mathbf{f}(\mathbf{Ze}_{s} + \mathbf{\alpha}) - \mathbf{f}(\mathbf{\alpha})
        \end{Vmatrix}_{2} \nonumber \\
        & -\begin{Vmatrix}
            \mathbf{MT}^{-1}
            \end{Vmatrix}_{2} \cdot \begin{Vmatrix}
            \Delta \mathbf{s}
        \end{Vmatrix}_{2}^{2}\label{eq:ve_dot2}\\
        &\le \begin{Vmatrix}
            \Delta \mathbf{s}
        \end{Vmatrix}_{2} \cdot \begin{Vmatrix}
            \mathbf{MT}^{-1}
            \end{Vmatrix}_{2} \cdot \begin{Vmatrix}
            \mathbf{C}_{s}\mathbf{Z}\Delta \mathbf{s}
        \end{Vmatrix}_{2} \nonumber \\
        &  -\begin{Vmatrix}
            \mathbf{MT}^{-1}
            \end{Vmatrix}_{2} \cdot \begin{Vmatrix}
            \Delta \mathbf{s}
        \end{Vmatrix}_{2}^{2} \label{eq:ve_dot3}\\
        &\le \qty(-\underline{\lambda} + \overline{\lambda}\begin{Vmatrix}
            \mathbf{C}_{s}\mathbf{Z}
        \end{Vmatrix}_{2})\begin{Vmatrix}
            \Delta \mathbf{s}
        \end{Vmatrix}_{2}^{2} \label{eq:ve_dot4},
    \end{align}
\end{subequations}

\noindent where \eqref{eq:ve_dot3} follows from the fact that $\mathbf{f}$ is a Lipschitz nonlinearity with Lipschitz constant $\mathbf{C}_{s}$. 
}
\end{proof}

\subsection{Supporting Analysis for Theorems \ref{thm:direct_power_inj} \& \ref{thm:bias}}
\label{subsec:appendix_proof}
Consider the smart inverter voltage measurement and power injection update dynamics given by:
\begin{equation}
    \mathbf{T}\mathbf{\dot{s}} = \mathbf{f}(\mathbf{Zs} + \mathbf{\bar{v}} + \mathbf{m}^{*}) -\mathbf{s} \label{eq:s_app2},
\end{equation}

\noindent where $\mathbf{m}^{*}$ is a constant vector. Assume that the system equilibrium lies in a region where the local Lipschitz constants of the inverter VV and VW functions meet the stability criteria of Proposition \ref{prop:direct} (i.e. Assumption \ref{assumption:direct2} holds).  Consider a system with equivalent dynamics but different states:
\begin{equation}
    \mathbf{T}\mathbf{\dot{s}}_{r} = \mathbf{f}(\mathbf{Zs}_{r} + \mathbf{\bar{v}} + \mathbf{m}^{*}) -\mathbf{s}_{r}  \label{eq:s_app_ref2},
\end{equation}

Let $\mathbf{e}_{s} = \mathbf{s} - \mathbf{s}_{r}$ denote the error between the two models \eqref{eq:s_app2} - \eqref{eq:s_app_ref2}.  The error dynamics can now be expressed as:
\begin{align}
    \mathbf{T}\dot{\mathbf{e}}_{s} &= - \mathbf{e}_{s} + \mathbf{f}(\mathbf{Ze}_{s} + \mathbf{Zs}_{r} + \mathbf{\bar{v}} + \mathbf{m}^{*}) \nonumber \\ & - \mathbf{f}(\mathbf{Zs}_{r} + \mathbf{\bar{v}} + \mathbf{m}^{*})  \label{eq:es_app2}.
\end{align}

\begin{proposition}
\label{prop:error}
Given the systems \eqref{eq:s_app2} - \eqref{eq:s_app_ref2}, and the associated error system of \eqref{eq:es_app2}, consider the function $V = \frac{1}{2}\mathbf{e}_{s}^\top\mathbf{M}\mathbf{e}_{s}$, where $M$ is positive definite and symmetric.  If Assumption \ref{assumption:direct2} holds for \eqref{eq:s_app2} - \eqref{eq:s_app_ref2}, then $\dot{V}$ is negative definite.
\end{proposition}

\begin{proof}
Let $\mathbf{\alpha} = \mathbf{Zs}_{r} + \mathbf{\bar{v}} + \mathbf{m}^{*}$.  The derivative of $V$ is:
\begin{subequations}
    \begin{align}
        \dot{V} &= \mathbf{e}_{s}^\top \mathbf{MT}^{-1} \Big( - \mathbf{e}_{s} + \mathbf{f}(\mathbf{Ze}_{s} + \mathbf{\alpha})  - \mathbf{f}(\mathbf{\alpha}) \Big) \label{eq:app_ve_dot1}\\
        &\le \begin{Vmatrix}
            \mathbf{e}_{s}
        \end{Vmatrix}_{2} \cdot \begin{Vmatrix}
            \mathbf{MT}^{-1}
            \end{Vmatrix}_{2} \cdot \begin{Vmatrix}
            \mathbf{f}(\mathbf{Ze}_{s} + \mathbf{\alpha}) - \mathbf{f}(\mathbf{\alpha})
        \end{Vmatrix}_{2} \nonumber \\
        & - \begin{Vmatrix}
            \mathbf{MT}^{-1}
            \end{Vmatrix}_{2} \cdot \begin{Vmatrix}
            \mathbf{e}_{s}
        \end{Vmatrix}_{2}^{2} \label{eq:app_ve_dot2}\\
        &\le \begin{Vmatrix}
            \mathbf{e}_{s}
        \end{Vmatrix}_{2} \cdot \begin{Vmatrix}
            \mathbf{MT}^{-1}
            \end{Vmatrix}_{2} \cdot \begin{Vmatrix}
            \mathbf{C}_{s}\mathbf{Z}\mathbf{e}_{s}
        \end{Vmatrix}_{2} \nonumber \\
        & -\begin{Vmatrix}
            \mathbf{MT}^{-1}
            \end{Vmatrix}_{2} \cdot \begin{Vmatrix}
            \mathbf{e}_{s}
        \end{Vmatrix}_{2}^{2}\label{eq:app_ve_dot3}\\
        &\le \qty(-\underline{\lambda} + \overline{\lambda}\begin{Vmatrix}
            \mathbf{C}_{s}\mathbf{Z}
        \end{Vmatrix}_{2})\begin{Vmatrix}
            \mathbf{e}_{s}
        \end{Vmatrix}_{2}^{2} \label{eq:app_ve_dot4}
    \end{align}
\end{subequations}

\noindent where $\underline{\lambda} = \lambda_{\min}(\mathbf{MT}^{-1})$ and $\overline{\lambda} = \lambda_{\max}(\mathbf{MT}^{-1})$.
\end{proof}


\bibliography{bibliog}

\begin{thebibliography}{10}
\providecommand{\url}[1]{#1}
\csname url@samestyle\endcsname
\providecommand{\newblock}{\relax}
\providecommand{\bibinfo}[2]{#2}
\providecommand{\BIBentrySTDinterwordspacing}{\spaceskip=0pt\relax}
\providecommand{\BIBentryALTinterwordstretchfactor}{4}
\providecommand{\BIBentryALTinterwordspacing}{\spaceskip=\fontdimen2\font plus
\BIBentryALTinterwordstretchfactor\fontdimen3\font minus
  \fontdimen4\font\relax}
\providecommand{\BIBforeignlanguage}[2]{{%
\expandafter\ifx\csname l@#1\endcsname\relax
\typeout{** WARNING: IEEEtran.bst: No hyphenation pattern has been}%
\typeout{** loaded for the language `#1'. Using the pattern for}%
\typeout{** the default language instead.}%
\else
\language=\csname l@#1\endcsname
\fi
#2}}
\providecommand{\BIBdecl}{\relax}
\BIBdecl

\bibitem{IEEE_1547}
\emph{{{IEEE} Standard for Interconnection and Interoperability of Distributed
  Energy Resources with Associated Electric Power Systems Interfaces}},
  Institute of Electrical and Electronics Engineers, IEEE 1547-2018, April
  2018.

\bibitem{rule21}
\emph{{Rule 21 Interconnection}}, Available:
  \url{https://www.cpuc.ca.gov/Rule21/}, California Public Utilities
  Commission, Std.

\bibitem{inverter2016}
B.~Seal, ``{Common Functions for Smart Inverters, 4th Ed.}'' Electric Power
  Research Institute, Tech. Rep. 3002008217, 2017.

\bibitem{jahangiri2013distributed}
P.~{Jahangiri} and D.~C. {Aliprantis}, ``Distributed volt/var control by pv
  inverters,'' \emph{IEEE Trans. Power Syst.}, vol.~28, no.~3, pp. 3429--3439,
  April 2013.

\bibitem{farivar2013equilibrium}
M.~Farivar, L.~Chen, and S.~Low, ``Equilibrium and dynamics of local voltage
  control in distribution systems,'' in \emph{Proc. IEEE Conf. Decis. Control},
  Dec 2013, pp. 4329--4334.

\bibitem{zhou2016local}
X.~{Zhou}, J.~{Tian}, L.~{Chen}, and E.~{Dall'Anese}, ``Local voltage control
  in distribution networks: A game-theoretic perspective,'' in \emph{2016 North
  American Power Symposium, NAPS 2016}, 2016, pp. 1--6.

\bibitem{braslavsky2017voltage}
J.~H. {Braslavsky}, L.~D. {Collins}, and J.~K. {Ward}, ``Voltage stability in a
  grid-connected inverter with automatic volt-watt and volt-var functions,''
  \emph{IEEE Trans. Smart Grid}, vol.~10, no.~1, pp. 84--94, 2019.

\bibitem{bakerNetwork2017}
K.~{Baker}, A.~{Bernstein}, E.~{Dall’Anese}, and C.~{Zhao},
  ``Network-cognizant voltage droop control for distribution grids,''
  \emph{IEEE Trans. Power Syst.}, vol.~33, no.~2, pp. 2098--2108, 2018.

\bibitem{eggli2020stability}
A.~{Eggli}, S.~{Karagiannopoulos}, S.~{Bolognani}, and G.~{Hug}, ``Stability
  analysis and design of local control schemes in active distribution grids,''
  \emph{IEEE Trans. Power Syst.}, vol.~36, no.~3, pp. 1900--1909, May 2021.

\bibitem{saha2020lyapunov}
S.~S. Saha, D.~Arnold, A.~Scaglione, E.~Schweitzer, C.~Roberts, S.~Peisert, and
  N.~G. Johnson, ``Lyapunov stability of smart inverters using linearized
  distflow approximation,'' \emph{IET Renew. Power Gener.}, vol.~15, no.~1, pp.
  114--126, 2021.

\bibitem{GHASEMI2016prevention}
M.~A. Ghasemi and M.~Parniani, ``Prevention of distribution network overvoltage
  by adaptive droop-based active and reactive power control of pv systems,''
  \emph{Electr. Power Syst. Res.}, vol. 133, pp. 313 -- 327, 2016.

\bibitem{sahoo2019cyber}
S.~Sahoo, T.~Dragi{\v{c}}evi{\'c}, and F.~Blaabjerg, ``Cyber security in
  control of grid-tied power electronic converters--challenges and
  vulnerabilities,'' \emph{IEEE Trans. Emerg. Sel. Topics Power Electron.},
  vol.~9, pp. 5326--5340, Oct. 2021.

\bibitem{practical2017}
W.~{Westerhof}, ``{Practical Proof - Horus Scenario},'' Available:
  \url{https://horusscenario.com/practical-proof/}, accessed: Jan. 2021.
  [Online].

\bibitem{russian2018}
``{Russian Government Cyber Activity Targeting Energy and Other Critical
  Infrastructure Sectors (Alert TA18-074A)},'' Available:
  \url{https://us-cert.cisa.gov/ncas/alerts/TA18-074A}, U.S. Cybersecurity \&
  Infrastructure Security Agency, accessed: Jan. 2021. [Online].

\bibitem{risks2019}
``{Risks Posed by Firewall Firmware Vulnerabilities},'' Available:
  \url{https://www.nerc.com/pa/rrm/ea/Lessons\%20Learned\%20Document\%20Library/20190901_Risks_Posed_by_Firewall_Firmware_Vulnerabilities.pdf},
  North American Electric Reliability Corporation, accessed: Jan. 2021.
  [Online].

\bibitem{spectrum2015}
P.~{Fairley}, ``{800,000 Microinverters Remotely Retrofitted on Oahu in One
  Day},'' Available:
  \url{https://spectrum.ieee.org/energywise/green-tech/solar/in-one-day-800000-microinverters-remotely-retrofitted-on-oahu},
  accessed: Jun. 2019. [Online].

\bibitem{astrom2008adaptive}
K.~Astrom and B.~Wittenmark, \emph{{Adaptive Control; 2nd ed.}}\hskip 1em plus
  0.5em minus 0.4em\relax Mineola, NY: Dover Publications, Inc., 2008.

\bibitem{singhal2019real}
A.~{Singhal}, V.~{Ajjarapu}, J.~{Fuller}, and J.~{Hansen}, ``Real-time local
  volt/var control under external disturbances with high pv penetration,''
  \emph{IEEE Trans. Smart Grid}, vol.~10, no.~4, pp. 3849--3859, 2019.

\bibitem{baran1989optimal}
M.~Baran and F.~F. Wu, ``Optimal sizing of capacitors placed on a radial
  distribution system,'' \emph{IEEE Trans. Power Del.}, vol.~4, no.~1, pp.
  735--743, Jan 1989.

\bibitem{baran1989network}
M.~E. {Baran} and F.~F. {Wu}, ``Network reconfiguration in distribution systems
  for loss reduction and load balancing,'' \emph{IEEE Trans. Power Del.},
  vol.~4, no.~2, pp. 1401--1407, 1989.

\bibitem{farivar2015local}
M.~{Farivar}, X.~{Zho}, and L.~{Chen}, ``Local voltage control in distribution
  systems: An incremental control algorithm,'' in \emph{2015 IEEE International
  Conference on Smart Grid Communications (SmartGridComm)}, 2015, pp. 732--737.

\bibitem{Khalil:1173048}
H.~K. Khalil, \emph{{Nonlinear systems; 3rd ed.}}\hskip 1em plus 0.5em minus
  0.4em\relax Upper Saddle River, NJ: Prentice-Hall, 2002.

\bibitem{nyserda}
{New York State Energy Research and Development Authority (NYSERDA)},
  ``{Mitigation Methods to Increase Feeder Hosting Capacity},'' Available:
  \url{https://www.nyserda.ny.gov/-/media/Files/Publications/Research/Electic-Power-Delivery/19-45-Mitigation-Methods-to-Increase-Feeder-Hosting-Capacity.pdf},
  {Prepared by Electric Power Research Institute (EPRI)}, Tech. Rep., 2019.

\end{thebibliography}
\bibliographystyle{IEEEtran}

\end{document}